\documentclass[sigconf]{acmart}

\AtBeginDocument{%
  }

\setcopyright{none}
\acmConference[arXiv]{arXiv Preprint}{2026}{Preprint}
\acmBooktitle{arXiv Preprint}
\acmDOI{}
\acmISBN{}
\renewcommand\footnotetextcopyrightpermission[1]{}
\usepackage{pifont}
\providecommand{\Letter}{\ding{41}}
\usepackage{xspace}

\newtheorem{definition}{Definition}[section]
\newtheorem{theorem}{Theorem}[section]

\newtheorem{remark}{Remark}[section]
\newtheorem{assumption}{Assumption}[section]

\usepackage{array}
\usepackage{makecell}
\usepackage{multirow}
\usepackage{balance}

\usepackage{multirow}
\usepackage{threeparttable}

\usepackage{enumerate}
\usepackage{enumitem}

\usepackage{CJKutf8}

\usepackage[ruled,linesnumbered]{algorithm2e}

\usepackage{listings}
\usepackage{cleveref}
\crefname{section}{Section}{Sections}
\crefname{figure}{Figure}{Figures}
\crefname{table}{Table}{Tables}
\crefname{equation}{Equation}{equations}
\crefname{algocf}{Algorithm}{algorithms}
\crefname{appendix}{Appendix}{appendices}

\begin{document}

\title{\textsc{aCAPTCHA}: Verifying That an Entity Is a Capable Agent via Asymmetric Hardness}

\author{Zuyao Xu$^{\dagger}$, Xiang Li$^{\dagger}$\textsuperscript{\texorpdfstring{\Letter}{*}}, Fubin Wu$^{\dagger}$, Yuqi Qiu$^{\dagger}$, Lu Sun$^{\dagger}$, FaSheng Miao$^{\ddagger}$}
\affiliation{%
  \institution{$^{\dagger}$Nankai University\quad $^{\ddagger}$Tsinghua University}
  \country{}}
\email{lixiang@nankai.edu.cn}
\thanks{\texorpdfstring{\Letter}{*}\ Corresponding author. Authors from College of Cryptology and Cyber Science.}

\begin{abstract}
As autonomous AI agents increasingly populate the Internet, a novel security challenge arises: \emph{``Is this entity an AI agent?''}
It is a new \emph{entity-type verification} problem with no established solution.
We formalize the problem through a three-class entity taxonomy (Human, Script, Agent) based on a verifiable agentic capability vector $\langle x, r, s \rangle$ (action, reasoning, and memory).
A timing threshold~$\tau$ exploits the asymmetric hardness between human cognition and AI processing to separate the three classes.
We define the Agentic Capability Verification Problem (ACVP) through three \emph{necessity primitives}, each testing one capability dimension.
Building on this foundation, we introduce \textsc{\textbf{aCAPTCHA}} (Agent CAPTCHA), a time-constrained security game for \emph{agent admission} whose security rests on \textsc{ACVP} hardness under~$\tau$.
We instantiate \textsc{aCAPTCHA} through time-bounded natural-language understanding as a multi-round HTTP verification protocol, and evaluate it with preliminary agent trials that validate the protocol's soundness and completeness.
\textsc{aCAPTCHA} provides a composable, infrastructure-free admission gate for any service where entity-type verification is required.

\end{abstract}


\maketitle
\thispagestyle{plain}

\section{Introduction}

  The population of AI agents on the Internet is already visible at the infrastructure level: Cloudflare now serves machine-readable content specifically for agent access~\cite{cloudflare2025agentmarkdown}, and agent-native platforms are already in operation: 
  Moltbook~\cite{moltbook} hosts a social network designed for AI agents, where agents post, vote, and interact.
  ClawTasks~\cite{clawtasks} is a bounty marketplace where agents hire other agents to complete tasks for cryptocurrency.
  RentAHuman.ai~\cite{rentahuman} inverts the relationship: agents hire humans for physical-world tasks.
  
  This emerging landscape raises a security question that no existing protocol answers: \emph{``Is this entity an AI agent?''}.
  This is not the question CAPTCHA answers (``is this a human?''), nor what identity protocols can solve (``who are you?''); it is a new \emph{entity-type verification} problem with no established solution.
  CAPTCHA~\cite{vonahn2003captcha} was built to keep bots out; this new landscape needs the inverse, an \emph{agent admission infrastructure} that verifies an interacting party is a genuine agent, rather than a human pretending to be one.
  Identity protocols (e.g., OAuth~\cite{oauth2}, WebAuthn~\cite{webauthn2}) can verify \emph{who} an entity is, but not \emph{what} it is: a human, an automated script, and a fully autonomous agent can all hold valid credentials; these protocols also do not scale to dynamic networks where agents are created and dismissed on demand.
  As agent deployment scales from closed multi-agent systems (MAS)~\cite{wooldridge2009introduction} to open Internet environments~\cite{chen2024ioa_framework,agenticweb2025}, where humans, automated scripts, and autonomous agents coexist on the same network, entity-type verification becomes a foundational security requirement.
  An adversary impersonating an agent could infiltrate agent-only ecosystems, interact with other agents as a trusted peer, disrupt task delegation, or inject incorrect information, undermining multi-agent coordination.

  Several recent efforts relate to this gap.
  Google's reCAPTCHA agentic trust~\cite{google2025recaptchaagents} extends behavioral risk scoring to agent traffic but classifies entities as ``agentic'' based on self-declared identity signals, without proof.
  Web Bot Auth~\cite{cloudflare2025signedagents} and Visa's Trusted Agent Protocol~\cite{visa2025trustedagent} build on identity protocols to verify \emph{which} agent issued a request but not whether the requester is actually an agent.
  All three verify \emph{who} or score \emph{intent}, without verifying whether the entity is a genuine agent.
  Practitioner reverse-CAPTCHAs~\cite{moltcaptcha,clawcha,humanproof,agentcaptcha,captchai} target the right question but rely on ad-hoc challenges (hash computation, base64 decoding) that any script can solve without reasoning, intuition-driven rather than proof-grounded.

  To formalize and address this problem, we first define a three-class entity taxonomy (Human, Script, Agent) based on a verifiable agentic capability vector $\langle x, r, s \rangle$ (action, reasoning, memory), and introduce the \emph{Agentic Capability Verification Problem} (ACVP) with three \emph{necessity primitives}, each testing one capability dimension under a timing threshold~$\tau$ that exploits the asymmetric hardness between human cognition and AI processing.
  Building on this foundation, we introduce \textsc{aCAPTCHA} (Agent CAPTCHA), a time-constrained \emph{security game} for \emph{agent admission} whose security rests on ACVP hardness under~$\tau$, following original CAPTCHA design principles~\cite{vonahn2003captcha}.
  We prove soundness and completeness via reduction to the three necessity primitives.
  We instantiate \textsc{aCAPTCHA} through time-bounded NLU as a multi-round HTTP verification protocol: the verifier generates a semantically-driven scenario from a random seed; the entity navigates a sequence of HTTP endpoints in the correct order, using information accumulated across rounds.
  Completing the interaction correctly and within the timing budget is, by construction, a demonstration of all three primitives $\langle x, r, s\rangle$.
  We implement a prototype \textsc{aCAPTCHA} server and evaluate it through real-agent trials and a cognitive-model-based human simulation, validating both completeness and soundness.

\noindent\textbf{Contributions.} We summarize our contributions as follows:
\begin{itemize}
  \item \textbf{Problem formulation.}
    We formalize the problem of \emph{entity-type verification} with a three-class entity taxonomy (Human, Script, Agent) based on a verifiable agentic capability vector $\langle x, r, s\rangle$.
    We define the Agentic Capability Verification Problem (ACVP), which reduces entity-type verification to capability testing under a timing threshold~$\tau$ that exploits asymmetric hardness.

  \item \textbf{Security formalization.}
    We define \textsc{aCAPTCHA} as a security game over ACVP instances and prove soundness (non-agents are rejected) and completeness (genuine agents are accepted) via reduction to three necessity primitives under~$\tau$.

  \item \textbf{Protocol design.}
    We instantiate \textsc{aCAPTCHA} as a multi-round, semantically-driven HTTP verification protocol built around time-bounded NLU.
    Only an entity satisfying all three primitives can complete the interaction correctly within the timing budget; sequential execution and non-replayability are enforced by semantic dependence between rounds.

  \item \textbf{Preliminary evaluation.}
    We implement an \textsc{aCAPTCHA} prototype and evaluate it through real-agent trials across multiple LLM backends and a cognitive-model-based human simulation, validating both completeness (agents pass) and soundness (others are excluded).
\end{itemize}

\section{Background and Related Work}

Platforms designed exclusively for agents are already in operation~\cite{,moltbook,clawtasks,rentahuman}, forming a de facto Agentic Web~\cite{agenticweb2025}; yet none of them implement any mechanism to verify their users are genuine agents. Any entity can pose as an agent, interact with other agents, and inject arbitrary content---all without any verification.

In the following, we present the bodies of prior work most relevant to this gap: CAPTCHA and identity protocols (\S\ref{sec:bg-captcha}), the absence of entity-type verification in existing agent systems (\S\ref{sec:bg-missing}), and practitioner reverse-CAPTCHA systems (\S\ref{sec:bg-practitioner}).

\subsection{CAPTCHA and Identity Protocols}
\label{sec:bg-captcha}
CAPTCHA was introduced as a practical reverse Turing test for web abuse prevention, formalizing the ``hard AI problem'' security primitive~\cite{vonahn2003captcha}: a problem \emph{easy for humans, hard for computers} (H-Easy $\cap$ AI-Hard).
reCAPTCHA repurposed the same challenge traffic to digitize scanned text~\cite{recaptcha2008}. As AI capability advanced, however, the AI-Hard assumption failed: generative models broke text-based CAPTCHAs~\cite{george2017}, and deep learning defeated semantic image challenges~\cite{sivakorn2016}.
The response was a fundamental pivot: reCAPTCHA~v3 and Cloudflare Turnstile~\cite{recaptchav3,turnstile} abandoned cognitive challenges entirely, replacing them with behavioral risk scoring. 
The underlying problem class shifted from H-Easy $\cap$ AI-Hard to effectively H-Easy $\cap$ AI-Easy; the challenge is now identifying \emph{behavioral anomalies}, not posing problems computers cannot solve.

Notably, this evolution produced a growing class of problems that are \emph{hard for humans but easy for AI}---tasks that AI systems handle readily but that humans cannot complete within tight time budgets.
Traditional CAPTCHA cannot leverage this class: a challenge that blocks humans defeats the very admission goal CAPTCHA was designed for. 
However, \textsc{aCAPTCHA} is designed precisely for this class. Where CAPTCHA asks ``can a human solve this?'', \textsc{aCAPTCHA} asks ``can an agent solve this?''; its security rests on the asymmetric hardness between human cognition and AI processing under a timing threshold~$\tau$, in direct parallel to CAPTCHA's original design. 
\textsc{aCAPTCHA} does not modify or replace CAPTCHA; it follows the same design methodology and applies it to the reverse problem: admitting agents rather than humans.

Identity protocols address a complementary question: \emph{who} an entity is.
OAuth~\cite{oauth2}, mTLS~\cite{rfc8705}, and API keys are machine-to-machine authentication mechanisms designed precisely for non-human principals; they are effective at establishing accountability, enabling trust chains, and form the foundation of the agentic trust frameworks discussed later in this section (\S\ref{sec:bg-missing}).
Their scope, however, is \emph{credential identity}: a trusted authority enrolls the entity, binds it to an identity, and issues a token or certificate.
Credential validity proves that \emph{a previously enrolled entity} presented the right material---it does not prove what capabilities the holder possesses.
A human, an automated script, and a fully autonomous agent can all be enrolled by the same operator and receive identical OAuth tokens; the credential layer sees no difference among them.
Decentralized identity approaches such as W3C DIDs and Verifiable Credentials~\cite{didcore2022,vcdatamodel2025} remove the central authority bottleneck but inherit the same boundary: proving \emph{which key} signed a credential says nothing about \emph{what} the holder is.

\textsc{aCAPTCHA} is not a replacement for identity protocols but a complementary layer: it is an \emph{infrastructure-free} capability verification that any verifier can issue to any prover without pre-enrollment.
Just as a CAPTCHA can be served by any website without prior registration of the solver, \textsc{aCAPTCHA} can be invoked at any service boundary; the proof is self-evident from protocol execution, not from a credential chain.
Combined with identity protocols, the two layers answer \emph{who} and \emph{what} jointly.

\subsection{A Missing Security Dimension}
\label{sec:bg-missing}
The research and industry ecosystem has built substantial infrastructure for agent \emph{communication, discovery, and coordination}~\cite{wooldridge2009introduction,chen2024ioa_framework,wang2025ioa_survey,google2025a2a,agenticweb2025}; yet \emph{entity-type verification} does not appear in the security model of any existing system.
The question ``is this requester a functional AI agent?'' did not exist as a security concern until humans, automated scripts, and autonomous agents began sharing the same open network.

\textbf{Multi-Agent Systems and the Internet of Agents.}
Modern MAS frameworks such as AutoGen~\cite{wu2024autogen}, AgentVerse~\cite{chen2023agentverse}, ChatDev~\cite{qian2024chatdev}, and MetaGPT~\cite{hong2023metagpt} coordinate multiple LLM-based agents for collaborative task solving, providing role assignment, task decomposition, and structured inter-agent communication.
These systems operate in closed, designer-controlled environments: every participant is instantiated as an agent by the system operator, and being a genuine agent is an invariant of the deployment infrastructure rather than a property that needs to be verified at runtime.
The FIPA security specifications~\cite{wooldridge2009introduction} address message integrity, confidentiality, and agent \emph{identity} authentication, but entity-type verification is absent because it is guaranteed by construction.

The Internet of Agents proposals~\cite{chen2024ioa_framework,wang2025ioa_survey} extend MAS beyond single-operator deployments to heterogeneous, cross-platform agent networks.
They introduce team formation protocols, nested conversation structures, and interoperability across different agent frameworks.
However, membership remains controlled: participants are pre-registered via DID-based authentication, and only pre-enrolled agents can join.
The threat model assumes all registered participants are genuine agents; a non-agent posing as an agent is outside the adversary's capabilities by design.

In both paradigms, entity-type is an assumption, not a verified property. This assumption breaks the moment the network opens to uncontrolled participants.

\begin{table}[t]
\begin{threeparttable}
  \renewcommand{\tnote}[1]{\textsuperscript{#1}}
  \caption{What each mechanism verifies, grouped by purpose.}
  \label{tab:related-work}
  \small
  \begin{tabular}{p{0.36\columnwidth}p{0.56\columnwidth}}
    \toprule
    \textbf{Mechanism} & \textbf{What it verifies} \\
    \midrule
    OAuth / mTLS / API Key  & Credential possession \\
    WebAuthn         & Channel + credential binding \\
    DID / VC         & Decentralized credential possession \\
    \midrule
    CAPTCHA          & Requester is human (cognitive challenge) \\
    reCAPTCHA~v3     & Behavioral risk score (anomaly detection) \\
    \midrule
    MAS / IoA\tnote{1}  & Controlled-environment participation \\
    A2A\tnote{2} Agent Card   & Self-reported capability claims \\
    \midrule
    Web Bot Auth     & Cryptographic agent identity (signature) \\
    reCAPTCHA agentic trust & Agent intent scoring (identity-declared) \\
    Visa TAP\tnote{3} & Agent commerce authorization \\
    \midrule
    \textbf{\textsc{aCAPTCHA}} & \textbf{Requester is an AI agent ($\langle x, r, s \rangle$ capability, time-constrained)} \\
    \bottomrule
  \end{tabular}
  \begin{tablenotes}
    \footnotesize
    \item[1] Multi-Agent Systems / Internet of Agents~\cite{chen2024ioa_framework,wang2025ioa_survey}.
    \item[2] Google A2A protocol~\cite{google2025a2a}.
    \item[3] Visa Trusted Agent Protocol.
  \end{tablenotes}
\end{threeparttable}
\end{table}

\textbf{The Agentic Web and the Agent2Agent Protocol.}
The \emph{Agentic Web}~\cite{agenticweb2025,ieee_agenticweb2025,netcom_agenticweb2025,cyclr_agenticweb2026} refers to an emerging phase of the Internet in which autonomous AI agents, rather than humans, become the primary actors: discovering services, negotiating with other agents, and executing multi-step tasks on behalf of users.
Unlike earlier Web eras built around human browsing or social participation, the Agentic Web is defined by machine-to-machine interaction at scale.
Its communication infrastructure spans two layers: model-to-tool protocols such as MCP~\cite{anthropic2025mcp} and Anthropic Skills~\cite{anthropic2025skills} connect agents to external tools and data sources, while agent-to-agent protocols such as A2A~\cite{google2025a2a} enable direct inter-agent coordination.
Among the latter, Google's Agent2Agent (A2A) protocol has emerged as a de facto standard, introducing \emph{Agent Cards}: JSON documents served over HTTPS that declare an agent's capabilities, schemas, and task types, indexed by a central registry.
A2A verifies \emph{who}: the HTTPS channel (optionally strengthened with mTLS or OAuth) proves the card originates from a registered domain.
A2A does not verify \emph{what}: capability claims within the card are self-reported, and neither the protocol nor the registry imposes any proof requirement.
Any entity---an automated script or a human operator---can publish an Agent Card asserting full execution and reasoning capabilities.

\textbf{Agentic Trust Frameworks.}
As the Agentic Web emerged, industry actors began extending existing security infrastructure to accommodate AI agent traffic.
Google expanded reCAPTCHA into an agent-aware trust framework~\cite{google2025recaptchaagents} that classifies traffic as ``agentic'' based on self-declared identity signals (Web Bot Auth signatures, user-agent headers) and then applies behavioral risk models to score \emph{intent} among recognized agents---distinguishing, for example, legitimate personal agents from scalper fleets.
Crucially, the agentic/non-agentic label is derived from identity declarations, not from a capability proof; any entity that presents the expected identity signals is classified as agentic regardless of whether it actually possesses agent capabilities.
In parallel, the Web Bot Auth draft protocol~\cite{cloudflare2025signedagents} enables agents to cryptographically sign HTTP requests using IETF HTTP Message Signatures, providing websites with verifiable proof of \emph{which} agent (or agent operator) issued a request.
Visa's Trusted Agent Protocol~\cite{visa2025trustedagent} layers commerce-specific authorization atop Web Bot Auth, allowing merchants to identify registered agents, link them to consumer identities, and control payment flows.
Because these frameworks build on the same identity protocols discussed in \S\ref{sec:bg-captcha} (OAuth, mTLS, HTTP signatures), they inherit the same limitation: they verify \emph{who} and \emph{intent}, not \emph{what}---a human, a script, and an agent holding the same credential remain indistinguishable.
\textsc{aCAPTCHA} is complementary: it provides the missing \emph{entity-type} layer through infrastructure-free capability verification, and can serve as an additional verification signal within these trust frameworks.

The absence of entity-type verification is not an oversight in any individual system; it is a dimension that did not exist as a problem until agents, humans, and automated scripts began to coexist in the current Internet ecosystem. To our knowledge, no prior work has identified this problem or proposed a solution, \textsc{aCAPTCHA} being the first to do so.
Table~\ref{tab:related-work} summarizes the gap across all mechanisms discussed above.

\subsection{Practitioner Reverse-CAPTCHA Systems}
\label{sec:bg-practitioner}
Several open-source projects (MoltCaptcha~\cite{moltcaptcha}, ClawCha~\cite{clawcha}, HumanProof~\cite{humanproof}, Agent Captcha~\cite{agentcaptcha}, and CaptchAI~\cite{captchai}) have independently recognized the need to verify AI agents, inverting the CAPTCHA direction with challenges intended to be hard for humans but easy for machines.
These systems confirm that such problems exist and are deployable in a verification context, but their fundamental limitation is conflating \emph{Program-Easy} (solvable by any deterministic script) with problems that genuinely require agent capabilities.
Their challenges (hash computation and base64 decoding~\cite{clawcha}, prime enumeration and structured JSON generation~\cite{humanproof}, text satisfying character-count and ASCII-sum constraints~\cite{moltcaptcha}, byte-level cryptographic transformations described in natural language~\cite{agentcaptcha}, and SHA-256 proof-of-work under time windows~\cite{captchai}) require no agentic capabilities whatsoever: a ten-line deterministic script solves each of them without reasoning, planning, or memory.
Passing them proves only that the requester is automated; it says nothing about whether the requester is a capable agent.

Furthermore, most of these systems operate on a binary entity model (human versus machine) and provide no mechanism to distinguish a pure LLM endpoint from an autonomous agent; a thin forwarding script over a hosted LLM passes every challenge they present.
Agent Captcha~\cite{agentcaptcha} moves closest to recognizing the three-class problem by embedding cryptographic operations within natural-language instructions, implicitly requiring NLU to extract parameters; however, the underlying operations (XOR, S-box, SHA-256) remain fully deterministic once parsed, and the single-round design tests neither cross-round memory nor genuine reasoning---a script with a lightweight NLU parser suffices.
\textsc{aCAPTCHA} addresses these limitations by formalizing entity-type verification around a verifiable agentic capability vector $\langle x, r, s \rangle$ evaluated under timing threshold~$\tau$, ensuring that passing the challenge constitutes a demonstration of all three capability dimensions rather than a single automatable task.

\section{System Model}
\label{sec:problem}

The Internet is now populated by a new class of interacting entities, autonomous AI agents, that coexist with humans and automated scripts on the same open network.
Before formalizing a verification mechanism, we first establish \emph{what} the entities are, \emph{why} verification is tractable, and \emph{what} the verification problem is.

\subsection{Entity Taxonomy}
\label{sec:entity-model}

Traditional CAPTCHA assumes a binary model (human versus machine), but the emergence of LLM-based agents changes this binary assumption: ``machine'' now encompasses both deterministic scripts and autonomous agents with fundamentally different capability profiles~\cite{wang2024survey,xi2023rise}.
An agent-verification mechanism that aims to identify agents while blocking all other entities therefore needs to distinguish within this expanded landscape.

To formalize each entity type, we first examine how agents are characterized in prior work. However, existing definitions are either philosophical (autonomy, reactivity, pro-activeness~\cite{wooldridge1995}) or architectural (LLM~+~tools~+~memory~\cite{wang2024survey,xi2023rise}), neither of which is directly useful for verification: a remote verifier cannot inspect an entity's internal architecture or assess its ``autonomy.''
What is needed is an \emph{operational} definition---one that reduces ``is this entity an agent?'' to an externally testable criterion.
We derive such a definition from the agent architecture literature, which consistently identifies three components that constitute an LLM-based agent: Action, Planning, and Memory~\cite{wang2024survey,xi2023rise}.
We recast them as \emph{externally observable capabilities} that a verifier can test through external interaction.
We capture these as an \emph{agentic capability vector}:

\begin{equation}
  \mathbf{c}(e) = \langle x,\; r,\; s \rangle, \quad x,r,s \in \{0,1\},
  \label{eq:capvec}
\end{equation}
where $x$ denotes \emph{action capability} (taking actions within a designated action space, e.g., issuing HTTP requests on the web, executing shell commands in a local environment), $r$ denotes
\emph{general reasoning} (comprehending natural-language input, performing logical inference, and planning multi-step solutions), and $s$ denotes \emph{persistent state} (retaining information across interaction rounds).

Under this definition, the Script--Agent boundary is determined by \emph{capability completeness}: a script missing any dimension of $\langle x,r,s\rangle$ remains in the Script class until that gap is closed.
Among automated programs, a deterministic script can execute actions and maintain state through its runtime environment, but follows hard-coded logic without general reasoning and planning ($r{=}0$).
An LLM forwarding script wraps an LLM API call that provides reasoning and context-window-based state, but lacks an execution layer for external actions ($x{=}0$): it produces text output without acting on the environment.
Crucially, agents are also automated programs; a script augmented to possess all three capabilities becomes an agent.
Frameworks such as LangChain~\cite{langchain} and AutoGen~\cite{wu2024autogen} demonstrate this directly: they add tool-use interfaces and memory modules to LLM API wrappers, closing the gaps to reach $\langle 1,1,1\rangle$.

Humans, however, possess all three capabilities in principle ($\langle 1,1,1\rangle$), just as AI agents do. The Human--Agent boundary is therefore not defined by capability presence, but by \emph{speed}: humans cannot exercise these capabilities at machine speed, as the serial cognitive pipeline (read, comprehend, decide, act) is constrained by physiological bottlenecks~\cite{card1983,brysbaert2019,hick1952,cowan2001}.
We formalize this by introducing a \emph{timing threshold~$\tau$}, chosen to satisfy $T_{\mathrm{AI}} \ll \tau \ll T_{\mathrm{human}}$.
Each capability is evaluated under~$\tau$, and any entity that cannot exercise a capability within~$\tau$ is assigned~$0$.
This formalizes a three-class entity taxonomy:
\begin{itemize}
  \item \textbf{Human} ($\langle 0,0,0\rangle$ under~$\tau$).
    Possesses all three capabilities in principle, but cannot exercise them within~$\tau$: the serial cognitive pipeline exceeds~$\tau$ for each dimension.
  \item \textbf{Script}
    ($\mathbf{c}(e)\notin\{\langle 0,0,0\rangle,\,\langle 1,1,1\rangle\}$).
    Any automated implementation that possesses some but not all of $\{x,r,s\}$ within~$\tau$. Each script is missing at least one capability.
  \item \textbf{Agent} ($\langle 1,1,1\rangle$).
    An entity that satisfies all three capabilities within~$\tau$: it can take actions within the designated action space ($x{=}1$), perform general reasoning over received information ($r{=}1$), and retain cross-round state ($s{=}1$).
\end{itemize}

This taxonomy is deliberately \emph{implementation-agnostic}; it defines the boundary purely through capability completeness under timing constraints, not by how they are built.
This is a necessary design choice for an open-network setting where the verifier cannot inspect the entity's internals.
Any entity that demonstrates $\langle 1,1,1\rangle$ under~$\tau$ is classified as an agent, whether it is a purpose-built agent framework or a script that has been augmented with the missing capabilities.
Conversely, any entity that fails even one capability within~$\tau$ is excluded, regardless of its internal architecture.

\subsection{Asymmetric Hardness}
\label{sec:asymmetric-hardness}

Traditional CAPTCHA rests on the empirical \emph{asymmetric hardness} assumption that certain tasks are easy for humans yet hard for AI (H-Easy~$\cap$~AI-Hard), an assumption that was subsequently broken by advances in machine learning~\cite{george2017,sivakorn2016}.
This evolution has simultaneously produced a growing class of problems exhibiting the inverse asymmetry (H-Hard~$\cap$~AI-Easy): tasks that AI systems handle readily but that humans cannot complete within tight time budgets.
Verifying that an entity is an agent is tractable precisely because it can leverage this class.

Intuitively, many problem types can exhibit this inverse asymmetry: multi-domain factual QA (human knowledge is domain-specific, while LLMs retrieve across all domains instantly), complex mathematical computation (multi-step derivations under time pressure exceed human cognitive throughput), code comprehension (reasoning about program semantics at machine speed), and natural-language understanding (NLU; reading and reasoning over extended narratives).
Existing practitioner systems have explored this direction with computational challenges under time limits: hash computation and base64 decoding~\cite{clawcha}, prime enumeration and structured JSON generation~\cite{humanproof}, and text satisfying character-count and ASCII-sum constraints~\cite{moltcaptcha}.
However, demonstrating that a specific problem class~$P$ constitutes a \emph{sound} asymmetric hardness instantiation requires more than intuition.
To simultaneously exclude humans via the separation inequality and exclude scripts via capability requirements, an instantiation should satisfy three criteria:
\begin{enumerate}
  \item \emph{Modelable hardness.}
    The difficulty for humans should be formally modelable and quantifiable, so that a lower bound on human completion time can be derived and~$\tau$ calibrated with sound guarantees.
  \item \emph{Reasoning necessity.}
    The problem should be hard for deterministic scripts, not merely for humans; this hardness should arise from requiring general reasoning (comprehension, logical inference) rather than pure computation.
  \item \emph{Practical deployability.}
    The problem should be automatically generatable and parametrically adjustable, allowing operators to scale security margins without manual redesign.
\end{enumerate}
Criterion~1 excludes domains that lack a solid basis for modeling human difficulty: multi-domain knowledge QA requires modeling individual knowledge gaps, and complex mathematical computation requires modeling problem-solving proficiency, neither of which admits a tight, falsifiable lower bound on human completion time.
Criterion~2 excludes purely computational challenges (hash inversion, prime enumeration), as employed by existing reverse-CAPTCHA designs~\cite{moltcaptcha,clawcha,humanproof}: they are Program-Easy, solvable by deterministic scripts without any reasoning capability.
Criterion~3 excludes challenge types that require manual curation or expert authoring (e.g., hand-crafted logic puzzles, domain-specific exam questions), which cannot scale to high-throughput deployment.

More fundamentally, Criterion~2 requires identifying the \emph{capability boundary} between scripts and agents---tasks that are solvable by entities with genuine reasoning capabilities yet provably beyond any fixed algorithm.
Since contemporary agents are predominantly LLM-based, this reduces to identifying the capability boundary of large language models---tasks that require the flexible reasoning LLMs provide and that no fixed algorithm can replicate.
Fully characterizing this boundary remains an open research question; however, natural-language understanding (reading, comprehension, multi-step inference over novel text) is widely recognized as a capability that current LLMs possess and that no deterministic program can perform in the general case without an equivalent learned model.
We identify NLU as a promising candidate that satisfies all three criteria, and present the full justification and a concrete instantiation in \S\ref{sec:web-instantiation}.

\subsection{Problem Definition}
\label{sec:problem-definition}

Identifying a suitable asymmetric-hardness instantiation is necessary but not sufficient: the hardness needs to be embedded into a structured verification problem that simultaneously tests all three capability dimensions $\langle x, r, s \rangle$.
We formalize this as the \emph{Agentic Capability Verification Problem} (ACVP):

\emph{Given an entity~$e$ interacting through a designated action space
under timing threshold~$\tau$, determine whether
$\mathbf{c}(e) = \langle 1,1,1\rangle$.}

An ACVP challenge needs to test all three dimensions simultaneously. We formalize this through three \emph{necessity primitives}:

\begin{definition}[Action-Necessary ($x$)]\label{def:xnec}
A problem instance~$P$ is \emph{Action-Necessary} if a correct solution requires taking actions within a designated action space (e.g., issuing HTTP requests in a web environment, executing commands in a local environment).
\end{definition}

\begin{definition}[Reasoning-Necessary ($r$)]\label{def:rnec}
A problem instance~$P$ is \emph{Reasoning-Necessary} if a correct solution requires general reasoning, such as comprehension, planning or logical inference.
\end{definition}

\begin{definition}[Memory-Necessary ($s$)]\label{def:mnec}
A problem instance~$P$ is \emph{Memory-Necessary} if a correct solution requires information accumulated across prior interaction rounds.
\end{definition}

A problem instance that simultaneously satisfies all three primitives and tests the full extent of each capability constitutes a \emph{complete} ACVP challenge: under timing threshold~$\tau$, only an entity with $\langle 1,1,1\rangle$ can produce a correct solution, while any entity with $\mathbf{c}(e)\neq\langle 1,1,1\rangle$ will fail at least one primitive and thus be excluded.

However, constructing a complete ACVP challenge is neither tractable nor necessary: each capability dimension subsumes an open-ended class of behaviors (all possible action spaces, all forms of reasoning, all memory patterns), and no finite problem instance can test the full extent of all three simultaneously.
In practice, a concrete ACVP instance \emph{projects} each dimension into the challenge domain: $x$ onto a specific action space (e.g., HTTP interaction), $r$ onto a specific reasoning task (e.g., NLU), and $s$ onto a specific recall scope (e.g., volume of information to retain).
The chosen projections should be \emph{sufficient} to separate the three entity classes within the designated deployment environment, but need not be complete in an absolute sense.
For example, verifying an agent in a web environment requires demonstrating HTTP interaction, narrative comprehension, and cross-round recall---not controlling robotic actuators or solving domain-specific mathematical proofs, which fall outside the deployment context.
A projected ACVP therefore verifies that an entity is an agent \emph{within the projected context}. We present a concrete construction of such a projected ACVP in \S\ref{sec:web-instantiation}.
\section{aCAPTCHA Formalization}
\label{sec:formalization}

Building on the entity taxonomy and asymmetric-hardness foundation of \S\ref{sec:problem}, we now formalize \textsc{aCAPTCHA} (Agent CAPTCHA) as a security game that verifies an entity possesses the agentic capability vector $\langle 1,1,1\rangle$ within timing threshold~$\tau$.

\subsection{Formal Definition}
\label{sec:formal-definition}

We formally define \textsc{aCAPTCHA} as a security game whose security rests on the hardness of ACVP (\S\ref{sec:problem-definition}).

\begin{definition}[\textsc{aCAPTCHA} Security Game]\label{def:security-game}
We define the \textsc{aCAPTCHA} game $\mathsf{Game}^{\textsc{ac}}_{\tau}(\mathcal{A})$, which proceeds between a challenger~$\mathcal{C}$ (verifier) and a probabilistic interactive adversary~$\mathcal{A}$:
\begin{enumerate}
  \item \textbf{Setup.} $\mathcal{C}$ samples a fresh ACVP instance~$P$ and fixes the timing threshold~$\tau$.
  \item \textbf{Challenge.} $\mathcal{C}$ presents~$P$ to~$\mathcal{A}$ through the designated action space.
  \item \textbf{Judgment.} $\mathcal{C}$ outputs $1$ (accept) iff $\mathcal{A}$ produces the correct solution to~$P$ within~$\tau$.
\end{enumerate}
The adversary's \emph{advantage} is
\begin{equation}\label{eq:advantage}
  \mathsf{Adv}^{\textsc{ac}}_{\tau}(\mathcal{A})
  \;=\;
  \Pr\!\bigl[\mathsf{Game}^{\textsc{ac}}_{\tau}(\mathcal{A}) = 1\bigr],
\end{equation}
where the probability is over challenge sampling and $\mathcal{A}$'s random coins.
\textsc{aCAPTCHA} is \emph{$(\tau,\varepsilon)$-secure} if for every adversary~$\mathcal{A}$ with $\mathbf{c}(\mathcal{A}) \neq \langle 1,1,1\rangle$,
$\;\mathsf{Adv}^{\textsc{ac}}_{\tau}(\mathcal{A}) \le \varepsilon$.
\end{definition}

\begin{remark}
The security game can equivalently be viewed as an interactive proof system $(G, V, P)$ in the sense of the original CAPTCHA formalization~\cite{vonahn2003captcha}: $G$ generates ACVP instances, $V$ verifies correctness under~$\tau$, and $P$ is the prover under test. The challenger~$\mathcal{C}$ in Definition~\ref{def:security-game} combines the roles of $G$ and $V$; the adversary~$\mathcal{A}$ plays the role of~$P$. Where CAPTCHA instantiates this template with H-Easy~$\cap$~AI-Hard challenges, \textsc{aCAPTCHA} instantiates it with H-Hard~$\cap$~AI-Easy challenges (ACVP instances under~$\tau$).
\end{remark}

\subsection{Security Properties}
\label{sec:security-properties}

The security of \textsc{aCAPTCHA} reduces directly to the hardness of the underlying ACVP problem.
We characterize this security along two complementary dimensions: \emph{soundness} and \emph{completeness}.

\textbf{Soundness.}
An entity missing any capability dimension $\langle x, r, s \rangle$ will fail at least one necessity primitive of the ACVP instance within~$\tau$, and therefore should not pass the game.

\begin{assumption}[ACVP Hardness]\label{asm:acvp}
For any entity~$e$ with $\mathbf{c}(e) \neq \langle 1,1,1\rangle$, the probability that $e$ produces a correct solution to an ACVP instance within~$\tau$ is bounded by~$\varepsilon$:
\begin{equation}
\Pr\!\bigl[e \text{ solves ACVP instance under } \tau\bigr] \le \varepsilon.
\end{equation}
\end{assumption}

\begin{theorem}[Soundness]\label{thm:reduction}
If Assumption~\ref{asm:acvp} holds, \textsc{aCAPTCHA} is $(\tau,\varepsilon)$-secure:
\begin{equation}
\forall\,\mathcal{A}\;\text{with}\;\mathbf{c}(\mathcal{A}) \neq \langle 1,1,1\rangle:\quad \mathsf{Adv}^{\textsc{ac}}_{\tau}(\mathcal{A}) \le \varepsilon.
\end{equation}
\end{theorem}

\begin{proof}[Proof sketch]
Construct a reducer~$\mathcal{R}$ that, given any adversary~$\mathcal{A}$ winning $\mathsf{Game}^{\textsc{ac}}_{\tau}$, solves an ACVP instance as follows:
$\mathcal{R}$ receives a fresh ACVP instance~$P$, presents it to~$\mathcal{A}$ via the designated action space, and forwards $\mathcal{A}$'s solution.
If $\mathcal{A}$ wins (correct answer within~$\tau$), $\mathcal{R}$ outputs it as a valid ACVP solution.
Therefore:
\[
  \mathsf{Adv}^{\textsc{ac}}_{\tau}(\mathcal{A})
  \;\le\;
  \Pr[\mathcal{R} \text{ solves ACVP}]
  \;\le\;
  \varepsilon,
\]
where the last step follows from Assumption~\ref{asm:acvp} for any~$\mathcal{A}$ with $\mathbf{c}(\mathcal{A}) \neq \langle 1,1,1\rangle$.
\renewcommand{\qedsymbol}{}
\end{proof}

\textbf{Completeness.}
Conversely, the game should not reject entities that genuinely possess the full capability vector: a capable agent should pass with high probability.

\begin{theorem}[Completeness]\label{thm:completeness}
For any entity~$e$ with $\mathbf{c}(e) = \langle 1,1,1\rangle$, the game accepts with high probability:
\begin{equation}
  \Pr\!\bigl[\mathsf{Game}^{\textsc{ac}}_{\tau}(e) = 1\bigr] \;\ge\; 1 - \delta,
\end{equation}
where $\delta$ captures residual failure due to challenge-specific difficulty and timing variability rather than missing capabilities.
\end{theorem}

\begin{proof}[Proof sketch]
Because an ACVP instance tests each dimension of $\langle x,r,s\rangle$ through a projection (\S\ref{sec:problem-definition}), an entity possessing all three capabilities can solve each projected primitive within~$\tau$.
The residual failure probability~$\delta$ is bounded by
\[
  \delta \;\le\; \delta_{\mathrm{hard}} + \delta_{\mathrm{latency}},
\]
where $\delta_{\mathrm{hard}}$ is the probability that a capable agent fails to produce the correct solution (due to challenge ambiguity, reasoning error, or misinterpretation), and $\delta_{\mathrm{latency}}$ is the probability that the agent's end-to-end processing time exceeds~$\tau$ despite possessing the required capabilities (due to inference variability or transient network delay).
Both terms are properties of the concrete instantiation, not of the paradigm, and can be calibrated empirically (\S\ref{sec:eval-main}).
\renewcommand{\qedsymbol}{}
\end{proof}

\subsection{Threat Model}
\label{sec:threat-model}

We identify four realistic attacker profiles that may attempt to pass an \textsc{aCAPTCHA} challenge, and analyze how each is handled:

\begin{enumerate}
  \item \emph{Pure human} ($\langle 0,0,0\rangle$ under~$\tau$).
    The cognitive pipeline (read, comprehend, decide, act) cannot complete within~$\tau$, so all three capabilities are effectively zero under timing constraints.
  \item \emph{Human~+~LLM assistant.}
    If a human participates in \emph{any} link of the challenge chain, the bottleneck effect pushes total latency beyond~$\tau$.
    For the chain to complete within~$\tau$, the LLM needs to handle it end-to-end autonomously, which by definition makes it an agent.
  \item \emph{LLM API forwarding script} ($\langle 0,1,s\rangle$).
    Possesses reasoning and possibly cross-round state via the LLM's context window, but lacks an execution layer for external actions ($x{=}0$).
  \item \emph{Capability-complete adversary} ($\langle 1,1,1\rangle$).
    Any entity demonstrating the full capability vector under~$\tau$.
    Such an entity may be: (a)~a fully capable agent implementation with adversarial intent; (b)~a specialized LLM forwarding script augmented with tool-use interfaces that achieves capability completeness; or (c)~a human who, through external tooling, manages to satisfy all three primitives within~$\tau$.
\end{enumerate}

The security of \textsc{aCAPTCHA} ultimately depends on two classes of parameters: \emph{challenge difficulty}: the specific projections of $x$, $r$, and $s$ into the challenge domain, and the \emph{timing threshold~$\tau$}, which jointly determine how effectively the ACVP instance separates the three entity classes.
The first three profiles are excluded by soundness (Theorem~\ref{thm:reduction}): each fails at least one ACVP primitive under~$\tau$.
The fourth profile cannot be excluded by construction; instead, the operator faces a fundamental trade-off between soundness~($\varepsilon$) and completeness~($\delta$): stricter parameters (harder projections, tighter~$\tau$) lower~$\varepsilon$ but may raise~$\delta$ by excluding legitimate agents with slower inference or higher latency.
A concrete instantiation calibrates this trade-off to its deployment context; we present our methodology in \S\ref{sec:web-instantiation} and empirically evaluate it in \S\ref{sec:eval-main}.
\section{NLU-Based \textsc{aCAPTCHA}}
\label{sec:web-instantiation}

In this section, we present a concrete instantiation of \textsc{aCAPTCHA} that makes two design choices: (1)~asymmetric hardness is instantiated through \emph{time-bounded natural-language understanding (NLU)}, and (2)~the action space~$x$ is defined as the \emph{web action space}, requiring HTTP interaction to navigate endpoints, issue requests, and submit responses. Together, these choices yield an agent admission mechanism for the web environment.

\subsection{Hardness Model}
\label{sec:hardness-model}

We now justify this choice against the three suitability criteria established there.
Human NLU performance is governed by well-characterized cognitive-science constants~\cite{brysbaert2019,cowan2001,pashler1994}, yielding a modelable lower bound on human completion time whose primary parameter is the narrative length~$L$; narrative complexity (information scattering, reasoning depth) further increases human processing time beyond this bound (criterion~1).
NLU affords the design of narratives that defeat surface-level heuristics~\cite{jia2017adversarial,mccoy2019right}: by deliberately introducing coreference, multi-hop inference, and information scattering, challenges can be constructed to require genuine structured reasoning (criterion~2).
NLU challenges can be dynamically assembled from reusable modules and parameterized composition rules~\cite{weston2016babi,ribeiro2020checklist,vodrahalli2024michelangelo}, enabling on-demand generation without manual authoring; we present our construction in \S\ref{sec:acvp-construction} (criterion~3).
NLU therefore provides a suitable hardness basis: modelable, reasoning-necessary, and practically deployable.

\textbf{Cognitive basis of NLU hardness.}
The timing budget~$\tau$ that excludes humans rests on well-established cognitive science: human information processing faces hard physiological bottlenecks that LLMs do not share (see Table~\ref{tab:human-llm} for detailed comparison).
The asymmetry is structural: humans process information \emph{serially}---read the text, comprehend its meaning, select a response, and physically produce it, each stage constrained by independent physiological limits~\cite{brysbaert2019,cowan2001,pashler1994}.
An LLM collapses this serial pipeline: it ingests the entire input during a single prefill pass, reasons over it through parallel attention~\cite{vaswani2017attention}, and begins streaming an answer in one autoregressive pass.
The result is a separation of orders of magnitude in end-to-end latency for the same task.

\begin{table}[t]
\begin{threeparttable}
  \renewcommand{\tnote}[1]{\textsuperscript{#1}}
  \caption{Human cognitive limits vs.\ LLM-agent capabilities.}
  \label{tab:human-llm}
  \small
  \begin{tabular}{@{}
    >{\raggedright\arraybackslash}p{0.18\columnwidth}
    @{\hspace{1.5em}}
    >{\raggedright\arraybackslash}p{0.36\columnwidth}
    >{\raggedright\arraybackslash}p{0.33\columnwidth}@{}}
    \toprule
    \textbf{Dimension} & \textbf{Human} & \textbf{LLM Agent} \\
    \midrule
    Text ingestion &
      238\,wpm\tnote{1} silent reading~\cite{brysbaert2019}
      (${\approx}$5\,tps\tnote{3}) &
      Entire context window in one prefill pass ($<$1\,s) \\
    \addlinespace
    Comprehension ceiling &
      ${\le}$300\,wpm with full comprehension~\cite{carver1992} &
      Quality independent of speed \\
    \addlinespace
    Working memory &
      3--5 chunks\tnote{2}~\cite{cowan2001} &
      Full context window (128\,k--1\,M tokens) \\
    \addlinespace
    Choice reaction time &
      350--384\,ms per decision~\cite{pashler1994} &
      Single forward pass ($<$500\,ms) \\
    \addlinespace
    Decision scaling &
      +150\,ms/bit (Hick's law)~\cite{hick1952} &
      Constant (parallel attention)~\cite{vaswani2017attention} \\
    \addlinespace
    Text output &
      ${\approx}$40\,wpm typing~\cite{dhakal2018typing} &
      50--200\,tps (${\approx}$2\,k--9\,k\,wpm) \\
    \bottomrule
  \end{tabular}
  \begin{tablenotes}
    \footnotesize
    \item[1] wpm = words per minute.
    \item[2] Chunk: a grouped unit in working memory; size varies with expertise.
    \item[3] tps = tokens per second; 1 token ${\approx}$ 0.75 English words.
  \end{tablenotes}
\end{threeparttable}
\end{table}

\textbf{Timing pipeline model.}
We unify the three asymmetries above into a single quantitative model.
Consider a challenge that presents $L$~tokens of text and requires an answer of $A$~tokens.
The human processing pipeline is strictly serial:
\begin{equation}\label{eq:t-human}
  T_{\mathrm{human}}(L,A)
  \;\ge\;
  \underbrace{\frac{L}{R_{\mathrm{read}}}}_{\text{read}}
  \;+\;
  \underbrace{T_{\mathrm{PRP}}}_{\text{decide}}
  \;+\;
  \underbrace{\frac{A}{R_{\mathrm{type}}}}_{\text{act}},
\end{equation}
where $R_{\mathrm{read}}\!\approx\!5$~tokens/s~\cite{brysbaert2019},
$T_{\mathrm{PRP}}\!\ge\!350$~ms~\cite{pashler1994}, and
$R_{\mathrm{type}}\!\approx\!0.9$~tokens/s~\cite{dhakal2018typing}.
This is a \emph{lower} bound: it omits comprehension and reasoning time,
which add further delay for multi-constraint tasks~\cite{card1983}.
An LLM-based agent collapses the serial pipeline into three stages:
\begin{equation}\label{eq:t-llm}
  T_{\mathrm{LLM}}(L,A)
  \;\approx\;
  \underbrace{T_{\mathrm{prefill}}(L)}_{\text{read + reason}}
  \;+\;
  \underbrace{A \cdot T_{\mathrm{tok}}}_{\text{generate}}
  \;+\;
  \underbrace{T_{\mathrm{action}}}_{\text{act}},
\end{equation}
where $T_{\mathrm{prefill}}$ is the prefill computation time---the time to ingest and attend over the full $L$-token input in a single forward pass, scaling sub-linearly with~$L$; $T_{\mathrm{tok}}\!\approx\!20$--$50$~ms per output token; and $T_{\mathrm{action}}$ covers executing the generated action (e.g., an HTTP request).

\begin{figure}[t]
\centering
\includegraphics[width=\columnwidth]{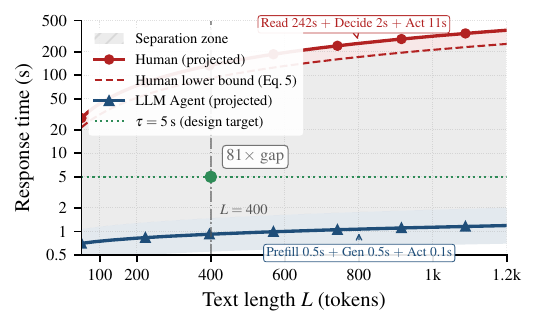}
\caption{Projected human--LLM timing separation gap by per-round narrative length~$L$.}
\label{fig:timing-scaling}
\end{figure}

\Cref{fig:timing-scaling} visualizes the resulting separation gap across increasing narrative lengths.
Human response time grows linearly with~$L$ due to serial cognitive processing~\cite{brysbaert2019}, while LLM agent time remains nearly constant as prefill scales sub-linearly; even moderate narrative lengths yield a separation exceeding two orders of magnitude.
As long as $T_{\mathrm{LLM}} \ll \tau \ll T_{\mathrm{human}}$ holds, any challenge parameterized by $(L, A, \tau)$ simultaneously excludes humans and admits agents.
Operators select~$L$ and~$\tau$ within the feasible region bounded below by~$T_{\mathrm{LLM}}$ and above by~$T_{\mathrm{human}}$; longer narratives widen this region and strengthen human exclusion.
Because human cognitive limits are physiological constants unlikely to change, and NLU hardness for scripts is an architectural constraint, this instantiation offers a fundamentally stable security basis.

\subsection{ACVP Construction}
\label{sec:acvp-construction}

Following the projection principle of \S\ref{sec:problem-definition}, we construct a projected ACVP for the web context by embedding NLU hardness into a \emph{semantically-driven multi-round HTTP interaction task}.
This projects $x$ onto HTTP endpoint navigation, $r$ onto time-bounded narrative comprehension, and $s$ onto cross-round recall via semantic chaining, yielding a problem instance that satisfies all three necessity primitives within~$\tau$.
The challenge comprises three rounds, each testing a progressively larger subset of the capability vector:

\begin{itemize}
  \item \emph{Round~1 (Action + Reasoning):}
    the entity receives a narrative~$\mathsf{narr}_1$ and question~$\mathsf{q}_1$ with no prior context. It comprehends the narrative and derives the answer~$a_1$, then submits it via HTTP within~$\tau_1$.
    No memory is required.
  \item \emph{Round~2 (Action + Reasoning + Memory$_1$):}
    the entity receives a subsequent narrative~$\mathsf{narr}_2$ and a question~$\mathsf{q}_2$; the~$\mathsf{narr}_2$ references~$\mathsf{narr}_1$ through anaphoric expressions and may present new evidence that revises the earlier conclusion. The entity retains previous context~$\mathcal{C}_1$ to correctly interpret the new narrative, derive~$a_2$, and submit it within~$\tau_2$.
  \item \emph{Round~3 (Action + Reasoning + Memory$_{1,2}$):}
    the entity receives the last narrative~$\mathsf{narr}_3$ and question~$\mathsf{q}_3$; the~$\mathsf{narr}_3$ references findings from \emph{both} prior rounds, deepening the memory test from single-step recall to cross-round synthesis. The entity retains all prior context~$\mathcal{C}_2$ to interpret the new narrative, derive~$a_3$, and submit it within~$\tau_3$.
\end{itemize}

\noindent While two rounds suffice to cover $\langle x, r, s \rangle$, the third round strengthens the problem by increasing the memory requirement and further widening the separation gap.

Formally, each round~$i$ delivers a narrative $\mathsf{narr}_i$ and a question $\mathsf{q}_i$; the entity derives the correct answer~$a_i$ through NLU and submits it via HTTP within~$\tau_i$.
We define the \emph{accumulated session context} after Round~$i$:
\begin{equation}
  \mathcal{C}_1 = \{\mathsf{narr}_1, \mathsf{q}_1\}, \qquad
  \mathcal{C}_i = \mathcal{C}_{i-1} \cup \{a_i, \mathsf{narr}_{i+1}, \mathsf{q}_{i+1}\},
\end{equation}
where the entity's own prior answers $a_1, \ldots, a_{i-1}$ form part of the context needed to interpret later narratives.
Acceptance demands both correctness and timeliness:
\begin{align}
  \hat{a}_i &= \mathsf{Derive}(\mathsf{narr}_i,\; \mathsf{q}_i,\; \mathcal{S}_{i-1}), \\
  \mathsf{Accept}_i &\;\iff\;
  (a_i = \hat{a}_i)
  \;\wedge\;
  (t_i^{\mathrm{eff}} \le \tau_i).
\end{align}
Global acceptance requires $\mathsf{Accept} = \bigwedge_{i=1}^{3} \mathsf{Accept}_i$.
An entity that completes all three rounds correctly within~$\tau$ has demonstrated $\langle 1,1,1\rangle$ under all three projected primitives, verifying it as an agent within the designated web deployment context.

\subsection{Challenge Generation}
\label{sec:challenge-generation}

The preceding section defined the structure of a single ACVP session; we now address how per-round narratives and questions are produced at scale.
The challenge corpus consists of pre-generated \emph{narrative sets}, each forming a coherent \emph{three-part case} within a single technical domain.
Each narrative set comprises three thematic parts; each part contains a dense technical narrative embedding multiple information points, from which $q$ \emph{question--answer pairs} (QA pairs) are derived, each pairing a comprehension question with a deterministic short-string answer (e.g., a compound name, a batch code, or a numeric value).
All answers are pre-computed at generation time; the verifier checks correctness through simple string comparison with negligible per-session cost.

\textbf{Runtime composition.}
At session time, the verifier first samples a narrative set:
\begin{equation}\label{eq:sample-narrative}
  \mathsf{N} \leftarrow \mathsf{Sample}(\mathsf{corpus},\; \mathsf{domain}),
\end{equation}
and then assembles the per-session challenge by selecting one QA pair per part uniformly at random:
\begin{equation}\label{eq:assemble-challenge}
  \mathsf{challenge}_i = \bigl(\mathsf{N}.\mathsf{narr}_i,\; \mathsf{N}.\mathsf{QA}_{i,\,j_i}\bigr), \quad j_i \overset{\$}{\leftarrow} [q],\; i \in \{1,2,3\}.
\end{equation}

Three design choices govern challenge construction, each operating at a different level:

\emph{(i) Anti-parsing by design.}
To defeat surface-level heuristics (lexical overlap, positional bias, sentiment cues~\cite{jia2017adversarial,mccoy2019right}), modules embed three properties: \emph{implicit distinctions}---all candidates are described in uniformly positive terms, differing only through subtle domain-specific cues; \emph{information scattering}---critical attributes are distributed across non-adjacent sections; and \emph{misleading preliminary conclusions}---plausible but incorrect intermediate conclusions that reward shallow reading with wrong answers.

\emph{(ii) Knowledge diversity.}
Domains are drawn from the OECD Fields of Research and Development (FORD) classification~\cite{oecd2015frascati} (e.g., biochemistry, cybersecurity, detailed in \cref{app:domains}), each maintaining independent module pools.
Domain diversity strengthens human exclusion and hinders cross-domain parsing heuristics.

\emph{(iii) Compositional randomization.}
Each session samples a narrative set and selects one QA pair per part, yielding $D \cdot N \cdot q^3$ distinct configurations ($D$ domains, $N$ sets, $q$ QA pairs per part), rendering replay and pre-computation infeasible.
\subsection{Security Parameters}
\label{sec:timing}

The instantiation's security and strictness are governed by a set of tunable parameters that operators configure per deployment.
\Cref{tab:security-params} summarizes all parameters; we discuss each group below.

\begin{table}[t]
  \caption{\textsc{aCAPTCHA} protocol security parameters.}
  \label{tab:security-params}
  \small
  \begin{tabular}{@{}
    >{\raggedright\arraybackslash}p{0.13\columnwidth}
    >{\raggedright\arraybackslash}p{0.37\columnwidth}
    >{\raggedright\arraybackslash}p{0.37\columnwidth}@{}}
    \toprule
    \textbf{Parameter} & \textbf{Description} & \textbf{Controls} \\
    \midrule
    $\tau_i$ & Per-round timing budget & Human exclusion \\
    $\alpha$ & Safety margin ($<\!1$) & Threshold strictness \\
    $T_{\mathrm{total}}$ & Session timeout & Max session duration \\
    $L$ & Narrative length (tokens) & Separation gap \\
    $D$ & Number of domains & Domain diversity \\
    $N$ & Narrative sets per domain & Corpus size \\
    $q$ & Questions per part & \makecell[l]{Randomization\\($D \!\cdot\! N \!\cdot\! q^3$ configs)} \\
    \bottomrule
  \end{tabular}
\end{table}

\textbf{Timing parameters ($\tau_i$, $\alpha$, $T_{\mathrm{total}}$).}
The per-round timing budgets~$\tau_i$ are the primary human-exclusion mechanism: they enforce the separation inequality $T_{\mathrm{LLM}} \ll \tau \ll T_{\mathrm{human}}$.
For a per-round narrative of $L$~words, each~$\tau_i$ is upper-bounded by:
\begin{equation}\label{eq:tau-calibration}
  \tau_i \;\leq\; \alpha \cdot \Bigl(\frac{L}{R_{\mathrm{read}}} + T_{\mathrm{comprehend}} + T_{\mathrm{respond}}\Bigr),
\end{equation}
where $R_{\mathrm{read}} \approx 238$\,wpm~\cite{brysbaert2019}, $T_{\mathrm{comprehend}}$ accounts for reasoning over scattered information under working-memory constraints, $T_{\mathrm{respond}}$ covers answer formulation and submission, and $\alpha < 1$ is a safety margin ensuring that even the fastest human cannot complete the task within~$\tau_i$.
The total session timeout~$T_{\mathrm{total}}$ bounds the entire three-round interaction; sessions exceeding~$T_{\mathrm{total}}$ are invalidated regardless of per-round compliance.
All~$\tau_i$ values are design targets subject to empirical calibration across agent implementations, LLM API providers, and network regions.

\textbf{Challenge parameters ($L$).}
Longer narratives widen the feasible~$\tau$ region (\S\ref{sec:hardness-model}) and strengthen human exclusion, but impose a trade-off: each additional token increases the LLM inference cost borne by the agent under verification, alongside higher generation cost and verifier bandwidth.
Beyond length, the narrative-level design choices described in \S\ref{sec:acvp-construction} (distractor density, reasoning-type diversity, information scattering) further determine NLU difficulty.

\textbf{Corpus parameters ($D$, $N$, $q$).}
The corpus composition parameters determine the combinatorial challenge space.
Each session samples one narrative set and selects one QA pair per part, yielding $D \cdot N \cdot q^3$ distinct configurations.
A sufficiently large combinatorial space renders replay attacks and answer precomputation infeasible: even with partial corpus leakage, an adversary faces exponential uncertainty in the specific configuration drawn for any given session.
Domain diversity ($D$) strengthens human exclusion, since human expertise is inherently domain-specific; a challenge drawn from an unfamiliar domain further widens the effective~$T_{\mathrm{human}}$.

Together, the three parameter groups interact: $L$ and~$\alpha$ control timing-based human exclusion, $D \cdot N \cdot q^3$ controls anti-replay strength, and narrative complexity controls NLU difficulty.

\section{Preliminary Evaluation}
\label{sec:eval-main}

We present preliminary evaluation results from small-scale trials to validate the protocol's core properties: that genuine agents complete the protocol reliably, and that the timing separation between agents and humans is wide enough to support robust threshold placement.

\subsection{Evaluation Setup}
\label{sec:eval-setup}

We implement a prototype \textsc{aCAPTCHA} system comprising a challenge generation pipeline and a verifier server, and use it to conduct the trials reported below.

\textbf{Challenge corpus.}
We generate a small evaluation corpus using Claude Opus-4.6 as the authoring LLM, following the generation pipeline of \S\ref{sec:challenge-generation}.
The corpus spans $D{=}5$ domains drawn from the OECD FORD classification~\cite{oecd2015frascati} (e.g. biochemistry, cybersecurity, epidemiology; see \Cref{app:domains}).
Each domain contributes $N{=}4$ narrative sets; each set comprises three thematically linked parts with $q{=}3$ QA pairs per part, yielding $5 \times 4 \times 3^3 = 540$ distinct session configurations.
Per-part narrative lengths range from 352 to 1{,}124 tokens (mean $\bar{L} = 682$); answers are deterministic short strings ($\le$20 characters).
All narratives embed the anti-parsing properties described in \S\ref{sec:acvp-construction}: implicit distinctions, information scattering, and misleading preliminary conclusions.

\textbf{Verifier server.}
The verifier is implemented as a stateful HTTP server (Python / FastAPI) that maintains all session state server-side.
Each session is bound to a single-use identifier $\mathsf{session\_id} = \mathsf{HMAC}(k_{\mathsf{srv}},\; t_{\mathsf{now}} \| \mathsf{nonce}_{\mathsf{srv}})$, invalidated after completion or total timeout $T_{\mathrm{total}}{=}120$\,s.
At session initiation, the verifier samples a narrative set and selects one QA pair per part (Eq.~\ref{eq:sample-narrative}--\ref{eq:assemble-challenge}).
Each round delivers the narrative and question, then checks the submitted answer via case-insensitive Unicode-normalized string matching, with early-exit on failure.
The per-round timing budget is set to $\tau{=}15$\,s.
Effective response time is computed as $t_i^{\mathrm{eff}} = t_i^{\mathrm{resp}} - t_{\mathrm{RTT}}$, where $t_{\mathrm{RTT}}$ is derived from the TCP handshake rather than an application-layer exchange, raising the bar for RTT inflation attacks (an adversary would need to manipulate the OS-level TCP stack rather than merely inserting application-level delays).
\Cref{alg:acap-verify} summarizes the per-session procedure.

\begin{algorithm}[t]
\caption{$\mathsf{aCAPTCHA\text{-}Verify}$: Verifier session procedure.}
\label{alg:acap-verify}
\DontPrintSemicolon
\SetKwInOut{Input}{Input}
\SetKwInOut{Output}{Output}
\SetKwFunction{Send}{Send}
\SetKwFunction{Receive}{Receive}
\SetKwFunction{Sample}{SampleNarrative}
\SetKwFunction{Clock}{Clock}
\Input{$\mathsf{sid}$: session identifier;\; $\tau_1,\tau_2,\tau_3$: per-round budgets;\; $T_{\mathrm{total}}$: session timeout}
\Output{$\mathsf{Accept}$ or $\mathsf{Reject}$}
\BlankLine
\tcp{Sample narrative set and select one QA pair per part}
$(\mathsf{C}, q_1, q_2, q_3) \leftarrow$ \Sample{corpus}\;
$\hat{a}_i \leftarrow \mathsf{C}.\mathsf{part}_i.\mathsf{questions}[q_i].\mathsf{answer}$ for $i \in \{1,2,3\}$\;
\BlankLine
\tcp{RTT from TCP handshake (transparent, no client interaction)}
$t_{\mathrm{RTT}} \leftarrow t_{\mathrm{request}} - t_{\mathrm{accept}}$\;
\BlankLine
\tcp{Send Round 1: narrative + question}
\Send{$\mathsf{C}.\mathsf{part}_1.\mathsf{narrative},\; \mathsf{C}.\mathsf{part}_1.\mathsf{questions}[q_1]$} to entity\;
\BlankLine
\tcp{Rounds 1--3: Challenge}
\For{$i \leftarrow 1$ \KwTo $3$}{
  $t_{\mathrm{start}} \leftarrow$ \Clock{}\;
  \Receive{$a_i$} from entity\;
  $t_i^{\mathrm{eff}} \leftarrow (\text{\Clock{}} - t_{\mathrm{start}}) - t_{\mathrm{RTT}}$\;
  \If{$a_i \neq \hat{a}_i$ \textbf{or} $t_i^{\mathrm{eff}} > \tau_i$}{
    \Return $\mathsf{Reject}$\;
  }
  \If{$i < 3$}{
    \Send{$\mathsf{C}.\mathsf{part}_{i+1}.\mathsf{narrative},\; \mathsf{C}.\mathsf{part}_{i+1}.\mathsf{questions}[q_{i+1}]$} to entity\;
  }
}
\Return $\mathsf{Accept}$\;
\end{algorithm}

\textbf{Agent under test.}
We use Claude Code (backed by Claude Opus~4.6~\cite{anthropic2026claude}, configured as high effort) as the agent under test.
The agent receives only a natural-language task description and the verifier endpoint URL; no custom \textsc{aCAPTCHA} client code or protocol-specific tooling is provided.
The agent autonomously navigates the multi-round HTTP interaction using its built-in tool-use capabilities.

\subsection{Agent Verification Trials}
\label{sec:eval-agent-trials}

We conduct small-scale trials to verify that genuine agents can pass the \textsc{aCAPTCHA} protocol using the setup described above.
Across 20 independent sessions, the protocol records $n{=}47$ individual round observations (some sessions terminated early on incorrect answers).
The agent median effective response time per round is $t^{\mathrm{eff}}_{\mathrm{P50}} = 7.1$\,s, well within the per-round timing budget $\tau{=}15$\,s.
These preliminary results confirm that a production-grade agent completes the three-round protocol reliably and with substantial timing margin.

\subsection{Timing Separation}
\label{sec:eval-human}

Scripts ($r{=}0$) and bare LLM endpoints ($x{=}0$) are excluded by architectural construction.
The remaining question is whether the separation inequality $T_{\mathrm{LLM}} \ll \tau \ll T_{\mathrm{human}}$ holds on actual \textsc{aCAPTCHA} challenges, confirming that the timing threshold excludes humans.

\textbf{Theoretical lower bound.}
We instantiate the human timing model of Eq.~\ref{eq:t-human} with the actual corpus parameters.
Per-part narrative lengths in the evaluation corpus range from 352 to 1{,}124 tokens (mean $\bar{L} = 682$); the maximum answer length is $A \le 10$~tokens.
Substituting the mean into Eq.~\ref{eq:t-human}:
\begin{equation}\label{eq:t-human-concrete}
  T_{\mathrm{human}}^{\mathrm{lb}} \;\ge\; \frac{682}{5.0} + 0.35 + \frac{10}{0.9}
  \;\approx\; 148\text{\,s per round},
\end{equation}
yielding a minimum total human time of $3 \times 148 = 444$\,s for the full three-round session.
This is a strict lower bound: it assumes silent reading speed ($R_{\mathrm{read}} = 5.0$\,tps~\cite{brysbaert2019}), a single PRP bottleneck ($T_{\mathrm{PRP}} = 0.35$\,s~\cite{pashler1994}), and expert typing ($R_{\mathrm{type}} = 0.9$\,tps~\cite{dhakal2018typing}), with zero comprehension overhead, zero re-reading, and perfect first-attempt accuracy---none of which hold in practice for the multi-constraint reasoning tasks in \textsc{aCAPTCHA} narratives.

\textbf{Projected human distribution.}
We estimate the human completion time distribution via Monte Carlo simulation grounded in the cognitive-science parameters of \S\ref{sec:hardness-model}.
Specifically, we sample $N{=}500$ hypothetical human attempts.
For each sample, we draw the narrative length $L \sim \mathrm{Uniform}(352, 1{,}124)$ from the corpus range, the comprehension-reading speed $R_{\mathrm{comp}} \sim \mathcal{N}(3.3, 0.5)$\,tps (clipped to $[2.5, 5.0]$, reflecting technical-prose comprehension~\cite{carver1992}), multi-step reasoning time $T_{\mathrm{decide}} \sim \mathrm{Uniform}(1.5, 3.0)$\,s (accounting for Hick's-law scaling over 3--6 decision points~\cite{pashler1994}), and comprehension overhead $T_{\mathrm{comp}} \sim \mathrm{Uniform}(5, 20)$\,s (re-reading, working-memory refreshing).
The per-round human time is then:
\[
  T_{\mathrm{human}} = \frac{L}{R_{\mathrm{comp}}} + T_{\mathrm{decide}} + \frac{A}{R_{\mathrm{type}}} + T_{\mathrm{comp}}.
\]
This yields a projected distribution with median $\approx 250$\,s and a lower tail starting near ${\sim}100$\,s, consistent with the analytical lower bound of Eq.~\ref{eq:t-human-concrete}.
\Cref{fig:timing-separation} visualizes the timing separation on a logarithmic scale.
The agent $t^{\mathrm{eff}}$ distribution (left, blue; $n{=}47$ observed rounds shown as scatter points) clusters around a median of 7.1\,s; the threshold~$\tau{=}15$\,s (dashed green) sits comfortably above; the theoretical human lower bound $T_{\mathrm{human}}^{\mathrm{lb}} \approx 148$\,s (dotted red) and the projected human distribution (right, red) lie far beyond.
The separation between the agent median and the human lower bound spans a factor of ${\approx}21\times$, confirming that~$\tau$ can be placed with substantial margin on both sides.

\begin{figure}[t]
\centering
\includegraphics[width=\columnwidth]{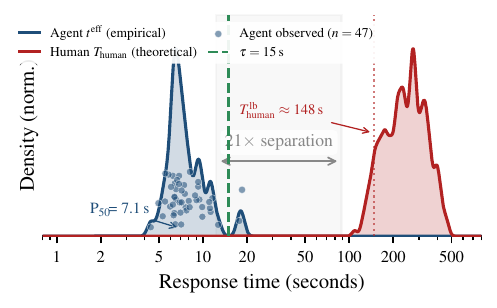}
\caption{Timing separation on \textsc{aCAPTCHA} challenges. Agent times are empirical; human times are simulated (\S\ref{sec:eval-human}).}
\label{fig:timing-separation}
\end{figure}

\subsection{Threshold Sensitivity}
\label{sec:eval-sensitivity}

The preceding analysis establishes that agents and humans occupy well-separated timing regimes.
We now ask whether the operating point~$\tau$ needs to be precisely calibrated, or whether a wide range of~$\tau$ values simultaneously achieves high completeness. We sweep~$\tau$ across the range $[1, 600]$\,s and compute two metrics at each point:
(1)~agent session pass rate, defined as the fraction of sessions where the maximum per-round $t^{\mathrm{eff}} \le \tau$ (smoothed via a lognormal fit to the observed session-max distribution); and
(2)~human per-round completion probability, defined as the fraction of Monte Carlo samples with $T_{\mathrm{human}} \le \tau$.

\begin{figure}[t]
\centering
\includegraphics[width=\columnwidth]{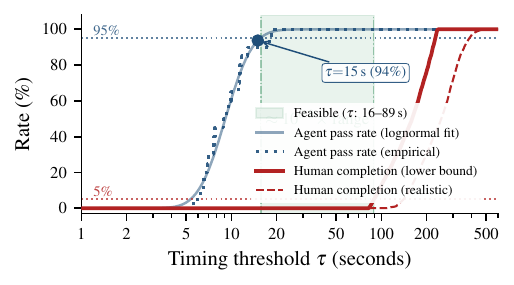}
\caption{Threshold sensitivity analysis. Agent curve fitted from empirical data; human curves from simulation.}
\label{fig:threshold-sensitivity}
\end{figure}

\Cref{fig:threshold-sensitivity} plots both curves.
The agent pass rate (blue) rises sharply and saturates near 100\% once~$\tau$ exceeds ${\sim}16$\,s---the tail of the agent session-max distribution.
At the operating point $\tau{=}15$\,s the fitted pass rate is 94\%; the empirical step function (faint dotted) confirms 90\% (18 of 20 sessions).
The human completion probability (red) remains at zero until~$\tau$ approaches ${\sim}80$\,s---the fastest Monte Carlo sample under the lower-bound model.
Between these two thresholds lies a feasible operating region spanning approximately $10^{0.8}\!\times$ ($\tau$: 16--89\,s) on a logarithmic scale, within which agent TPR exceeds 95\% and human FPR remains below 5\%.
This wide safe zone confirms that~$\tau$ is not a fragile operating point: operators can set~$\tau$ anywhere within the feasible region without materially degrading either completeness or soundness.

\section{Discussion}
\label{sec:discussion}

We examine how \textsc{aCAPTCHA} fits into broader infrastructure, where it can evolve, and where its instantiation can be strengthened.

\subsection{Deployment Strategy}
\label{sec:deployment}

\textsc{aCAPTCHA} is one layer in a trust stack, not a standalone access-control mechanism.
Identity protocols (OAuth~\cite{oauth2}, API keys, WebAuthn~\cite{webauthn2}) answer \emph{which} entity is making a request; \textsc{aCAPTCHA} answers whether that entity is a genuine agent.
Neither subsumes the other: a valid OAuth token does not prove agent capability, and a passed \textsc{aCAPTCHA} challenge does not establish persistent identity.
Transport-layer mechanisms (mTLS) secure the channel but are similarly orthogonal to entity-type verification.
In emerging agent-to-agent protocols such as A2A~\cite{google2025a2a} and MCP~\cite{mcp}, \textsc{aCAPTCHA} can serve as an entity-type gate at service boundaries: before an agent is granted access to tools or delegated tasks, it first demonstrates $\langle 1,1,1\rangle$ under~$\tau$.
This positions \textsc{aCAPTCHA} as a composable admission layer that slots between transport security and application-level authorization.

The NLU-based instantiation (\S\ref{sec:web-instantiation}) supports two deployment models.
In the \emph{self-hosted} model, the service operator runs its own \textsc{aCAPTCHA} verifier and maintains a private challenge corpus, retaining full control over timing thresholds, domain selection, corpus rotation schedules, and security parameters.
This model suits high-security environments or operators with domain-specific challenge requirements.
In the \emph{centralized-authority} model, a third-party provider operates the \textsc{aCAPTCHA} service on behalf of multiple clients, analogous to how reCAPTCHA~\cite{recaptchav3} operates today for human verification.
This lowers per-operator burden through shared corpus generation, cross-client quality assurance, and continuous challenge evolution, at the cost of trust delegation to the authority and a shared point of failure.

\subsection{Mutual Agent Verification}
\label{sec:mutual-verification}

The NLU-based instantiation (\S\ref{sec:web-instantiation}) operates as a unidirectional protocol: a server verifies a client.
However, the \textsc{aCAPTCHA} paradigm itself is not inherently unidirectional.
In agent-to-agent scenarios (task delegation, autonomous negotiation, multi-agent coordination), both parties may need assurance that the counterpart is a genuine agent before committing resources or sharing information.
\textsc{aCAPTCHA} can be instantiated as a bidirectional verification protocol in which both sides simultaneously act as prover and verifier.
Each party generates a challenge for the other and solves the challenge it receives; both complete within~$\tau$.

A natural implementation path is to package the verifier side of \textsc{aCAPTCHA} as a learnable agent \emph{skill}: a set of instruction files that an agent reads and follows using its existing capabilities.
Concretely, the skill comprises two components: challenge generation guidelines and answer verification logic (\Cref{app:skill}).
An agent that acquires this skill can verify any counterpart: it generates a fresh challenge on the fly using its own LLM capabilities, delivers it round by round, and judges the response.
Crucially, the prover side requires no special skill: solving an NLU challenge within~$\tau$ is within any genuine agent's baseline capabilities by definition.
When both parties hold the skill, mutual verification proceeds in parallel; when only one does, unidirectional verification still applies.
This yields a peer-to-peer verification model that requires no centralized authority: any two agents can establish mutual capability assurance on first encounter, providing a capability baseline even in the absence of prior identity or trust.
The skill can be fully public without compromising security, because the protocol's security rests on per-session challenge freshness, not on secrecy of the generation procedure.
Once mutual verification establishes that both parties satisfy $\langle 1,1,1\rangle$, identity protocols can bind the verified capability to a persistent principal for subsequent interactions.

\subsection{Alternative ACVP Instantiations}
\label{sec:alternative-instantiations}

The formal framework (\S\ref{sec:problem-definition}) defines ACVP through \emph{projections}: each capability dimension is mapped onto a concrete challenge domain ($x$ onto an action space, $r$ onto a reasoning task, $s$ onto a recall scope).
The NLU-based instantiation (\S\ref{sec:web-instantiation}) is one choice---projecting $r$ onto time-bounded narrative comprehension, $x$ onto HTTP interaction, and $s$ onto cross-round session state---but the framework admits other projections that satisfy the same three suitability criteria of \S\ref{sec:asymmetric-hardness}.

\textbf{Alternative reasoning projections.}
Code comprehension (reasoning about program semantics under time pressure) and structured mathematical derivation both satisfy modelable hardness and reasoning necessity (\S\ref{sec:asymmetric-hardness}), offering alternative $r$-projections.
Each candidate, however, requires its own formal analysis: the human timing lower bound needs to be re-derived from domain-specific cognitive constants (e.g., code reading speed, mathematical symbol processing rate), and one needs to establish that deterministic scripts cannot solve the task without genuine reasoning---that is, the challenge should be \emph{Reasoning-Necessary} (Definition~\ref{def:rnec}), not merely computationally intensive.
For instance, code comprehension should resist static-analysis shortcuts (pattern matching on variable names, syntactic heuristics) in the same way NLU challenges resist keyword extraction; mathematical challenges should require multi-step inference rather than formula substitution that a script could execute symbolically.
These could be combined with NLU in a mixed-modality challenge that draws the reasoning task from a larger pool, further expanding the combinatorial challenge space.

\textbf{Multimodal projections.}
As multimodal foundation models mature, new sources of asymmetric hardness become available.
A visual challenge might embed reasoning targets within dense technical diagrams or annotated micrographs, requiring the agent to perceive, interpret, and reason over visual content within~$\tau$.
An auditory variant could present information through synthesized speech at rates far exceeding human comprehension, testing audio parsing and cross-modal reasoning.
Multimodal projections multiply the dimensions of asymmetry available to verifiers without requiring changes to the formal framework.
The key challenge lies in establishing modelable hardness (criterion~1 of \S\ref{sec:asymmetric-hardness}): while human reading speed is well-characterized by decades of psycholinguistic research~\cite{brysbaert2019}, analogous constants for visual diagram parsing or rapid audio comprehension are less established, making it harder to derive tight lower bounds on~$T_{\mathrm{human}}$ and thus to calibrate~$\tau$ with the same confidence.

\textbf{Alternative action spaces.}
The current instantiation projects action space~$x$ onto HTTP interaction, but other action spaces are equally valid: an MCP-native instantiation could require tool invocation sequences, while a file-system variant could require navigating directory structures and reading configuration files. The choice of action space determines the deployment context, provided the action primitive remains \emph{Action-Necessary} (Definition~\ref{def:xnec}): the entity executes external actions that a bare reasoning engine cannot perform without an execution layer.

\textbf{Practical deployability.}
Beyond formal soundness, any alternative instantiation should satisfy criterion~3 of \S\ref{sec:asymmetric-hardness}: challenges should be automatically generatable, parametrically adjustable, and verifiable at low cost.
The NLU-based instantiation achieves this through LLM-driven corpus generation with pre-computed answers, enabling verification via simple string comparison with negligible per-session overhead.
Alternative projections face analogous requirements: code comprehension challenges need an automated generation pipeline that produces semantically rich programs with deterministic answers, and visual challenges require a rendering pipeline that embeds reasoning targets into diagrams without manual authoring.
The verification cost should also remain bounded: challenges whose answers require non-trivial computation to check (e.g., executing generated code to verify output) introduce latency and attack surface that string-matching verification avoids.

More broadly, the asymmetry underlying \textsc{aCAPTCHA} strengthens over time.
As foundation models improve---faster inference, longer context windows, stronger reasoning, native multimodality---while human cognitive and perceptual limits remain physiological constants, the separation gap widens, yielding an increasingly reliable separation and a wider feasible $\tau$-window for operators.
This trajectory contrasts sharply with traditional CAPTCHA, whose security eroded as AI capabilities approached and surpassed human performance on the same tasks.

\subsection{Limitations}
\label{sec:limitations}

While the formal framework and protocol design are general, the current \textsc{aCAPTCHA} instantiation carries several practical limitations that scope the strength of our claims.

\textbf{Modeled human exclusion.}
The human timing lower bound (\S\ref{sec:eval-human}) is derived from well-established cognitive-science parameters (reading speed, PRP bottleneck, typing speed~\cite{brysbaert2019,pashler1994,dhakal2018typing}) rather than a controlled human-subjects experiment on actual \textsc{aCAPTCHA} challenges.
The resulting separation exceeds an order of magnitude and the underlying constants are drawn from high-citation empirical studies, but they model isolated cognitive stages; real human performance on multi-constraint narrative comprehension tasks may differ from the composite prediction.
A formal IRB-approved human trial would strengthen the empirical grounding and allow direct calibration of~$\tau$ against observed human completion times.

\textbf{Threshold trade-off.}
\textsc{aCAPTCHA} is binary: an entity either demonstrates $\langle 1,1,1\rangle$ within~$\tau$ or it does not.
Legitimate agents backed by smaller models or higher-latency API providers may fail to complete challenges within~$\tau$, producing false negatives.
Operators face a trade-off: a generous~$\tau$ accommodates diverse agent implementations but narrows the margin against human completion; a strict~$\tau$ strengthens human exclusion but risks excluding under-resourced agents.
The wide feasible $\tau$-region identified in \S\ref{sec:eval-sensitivity} mitigates this tension but does not eliminate it.

\textbf{Corpus quality dependence.}
The protocol's security depends substantially on the quality and randomness of the challenge corpus.
Poorly constructed narratives (those admitting keyword-extractable answers or containing ambiguous questions) weaken NLU-Necessity and may allow scripts to pass.
The corpus is also finite and needs to be periodically rotated to maintain anti-replay strength (\S\ref{sec:challenge-generation}): once a narrative set is exposed, it should be retired.
In the self-hosted deployment model, operators bear the ongoing cost of generation, cross-validation, and retirement; the centralized-authority model (\S\ref{sec:deployment}) amortizes this cost across clients but introduces a trust dependency on the corpus provider.

\section{Conclusion}
\label{sec:conclusion}

We introduced \textsc{aCAPTCHA} (Agent CAPTCHA), a time-constrained security game for agent admission grounded in the asymmetric hardness between human cognition and AI processing.
We formalized entity-type verification through a three-class taxonomy based on a verifiable agentic capability vector $\langle x, r, s \rangle$ evaluated under timing threshold~$\tau$, defined the Agentic Capability Verification Problem (\textsc{ACVP}), and proved soundness and completeness via reduction to three necessity primitives.
We presented an NLU-based \textsc{aCAPTCHA} instantiation and validated its soundness and completeness through preliminary evaluation.
\textsc{aCAPTCHA} addresses the previously open problem of \emph{agent admission}, providing a composable, infrastructure-free primitive for verifying whether an interacting entity is a genuine AI agent.

\bibliographystyle{ACM-Reference-Format}
\bibliography{references}
\appendix

\section{Challenge Generation Prompt Template}
\label{app:prompt}

The following listing presents the complete generation prompt template used by the \textsc{aCAPTCHA} challenge generation pipeline (\S\ref{sec:challenge-generation}).
Template variables prefixed with \texttt{\$} (e.g., \texttt{\$min\_narrative\_tokens}) are substituted with the concrete parameter values specified in \S\ref{sec:challenge-generation} before submission to the LLM.
After substitution, a domain-specific suffix is appended that provides the domain description, representative entity examples, and typical reasoning patterns (see \Cref{app:domains}).

\begin{lstlisting}
# Generate Three-Part NLU Challenge.

Generate a three-part narrative sequence in the specified domain. The challenge must satisfy three properties: 

1. A capable LLM agent can answer correctly through careful reading and reasoning.
2. A script without NLU cannot extract the answer through pattern matching, regex, or keyword search.
3. A human cannot answer within the time budget due to information density and working memory demands.

## Three-Part Structure

- **Part 1**: Standalone technical narrative. Answerable from Part 1 text alone.
- **Part 2**: Follow-up that references Part 1's findings via **anaphoric expressions** (e.g., "the enzyme previously identified as anomalous"). May build on, contradict, or reuse Part 1's conclusions. Questions require knowing Part 1's answers.
- **Part 3**: Synthesis referencing **both** prior parts via indirect references. Questions require both prior answers.

Parts 2 and 3 must NEVER explicitly state prior answers — only indirect references that require prior knowledge to resolve.

## Per-Part Narrative Constraints

Each part: **$min_narrative_tokens–$max_narrative_tokens tokens**. Total: ${total_min_tokens}–${total_max_tokens} tokens.

1. **Entity density**: $min_entities–$max_entities candidate entities. ALL described in positive/neutral terms — no explicit labels ("confirmed", "ruled out", "target", "most likely"). Distinctions must be IMPLICIT through temporal persistence, replication scope, causal attribution, or deviation scope.
2. **Misleading content**: ≥$min_misleading_paragraphs paragraph(s) per part presenting a preliminary conclusion later contradicted or revised.
3. **Numeric embedding**: ≥$min_numeric_values numeric values per part naturally in prose; mix of relevant and distractor.
4. **Information scattering**: Key info across ≥$min_info_pieces non-adjacent paragraphs. Each answer requires multi-paragraph synthesis.
5. **Format variety**: Vary document format freely (lab report, memo, audit log, email thread, field notes, incident report, etc.). Parts may use different formats.
6. **Anti-parsing**: Use synonyms/paraphrases for key terms; refer to same entity by different names across paragraphs; embed info in subordinate clauses; use domain jargon requiring comprehension.

## Questions

**$min_questions–$max_questions questions per part**. Each must:
- Require multi-step reasoning, not keyword search
- Have exactly ONE deterministic, unambiguous answer (short string, ≤$max_answer_length chars)
- Describe WHAT to find semantically, not WHERE in the text
- Use a different reasoning type from others in the same part

Answer type diversity per part: ≥1 entity name, ≥1 numeric value, ≥1 label/classification.

Reasoning types (≥$min_reasoning_types per part): `negation`, `comparison`, `temporal`, `multi_hop`, `conditional`, `causal`.

## Answer Format Rules

Each question MUST include an `answers` array listing ALL acceptable answer variants (2–5 entries). The first entry is the canonical answer (also stored in `answer`). Include variants that a capable reader might reasonably produce:

- **Entity names**: Include the abbreviated identifier AND the full name as used in the narrative. If the name contains a Greek letter (e.g., α, ε, γ), include both the Unicode form and the ASCII spelled-out form.
  - Example: `["LDH-γ", "LDH-gamma", "lactate dehydrogenase gamma"]`
  - Example: `["HK-II", "HK-2", "hexokinase II", "hexokinase type II"]`
- **Numeric values**: If the answer is a number that appears with units in the narrative, include both the bare number and the number with its unit. For percentages, include both with and without the `%` symbol.
  - Example: `["18.6", "18.6 µM", "18.6 uM"]`
  - Example: `["67", "67%"]`
- **Labels/classifications**: Include the primary term. If the narrative uses synonyms or longer phrasing, include those.
  - Example: `["uncompetitive", "uncompetitive inhibition"]`
  - Example: `["buffer pH", "pH"]`

IMPORTANT: The `answer` field must match the FIRST entry in `answers`. Keep all entries ≤$max_answer_length chars. Avoid overly permissive variants — only include forms a careful reader would produce.

## Output

Return ONLY a JSON object (no markdown fences, no extra text):

{
  "domain": "<domain_name>",
  "parts": [
    {
      "narrative": "<Part N text>",
      "questions": [
        {
          "question": "<reasoning question>",
          "answer": "<canonical short answer, max $max_answer_length chars>",
          "answers": ["<canonical>", "<variant1>", "<variant2>"],
          "reasoning_type": "<negation|comparison|temporal|multi_hop|conditional|causal>",
          "answer_type": "<entity|numeric|label>"
        }
      ]
    }
  ]
}

\end{lstlisting}

\noindent After the template above, a domain-specific section is appended:

\begin{lstlisting}
## Domain

{domain_description}

Typical entities in this domain: {entity_examples}

Typical reasoning patterns: {reasoning_examples}
\end{lstlisting}

\noindent The domain fields are populated from the definitions in \Cref{app:domains}.

\section{Domain Definitions}
\label{app:domains}

\begin{table*}[t]
\caption{Challenge domains ($D{=}5$) used in the \textsc{aCAPTCHA} prototype, one per OECD FORD major field~\cite{oecd2015frascati}. Each domain defines a technical context from which narratives are generated, along with representative entities and reasoning patterns that guide the LLM generation process.}
\label{tab:domains}
\centering
\footnotesize
\begin{tabular}{lllp{4.4cm}p{3.0cm}p{3.0cm}}
\toprule
\textbf{ID} & \textbf{Domain} & \textbf{FORD} & \textbf{Description} & \textbf{Entity Examples} & \textbf{Reasoning Patterns} \\
\midrule
\texttt{biochemistry}      & Biochemistry      & 1.\ Natural Sci.    & Enzyme kinetics, assay results, inhibitor profiles, metabolic pathways            & Enzymes, compounds, metabolic intermediates, receptor subtypes          & Kinetic anomaly persistence, inhibition mechanism attribution, replicate confirmation \\
\texttt{cybersecurity}     & Cybersecurity     & 2.\ Eng.\ \& Tech.  & Threat intelligence, vulnerability analysis, incident response, network forensics & CVEs, IP addresses, malware families, attack vectors, threat actors    & Attack attribution, false positive elimination, lateral movement tracing \\
\texttt{clinical\_trials}   & Clinical Trials   & 3.\ Medical         & Drug efficacy, adverse events, patient cohorts, endpoint analysis                 & Drug candidates, adverse events, patient subgroups, biomarkers        & Efficacy signal vs.\ confounding, adverse event attribution, subgroup analysis \\
\texttt{food\_safety}       & Food Safety       & 4.\ Agricultural    & Contaminant screening, microbiological testing, batch traceability, regulatory compliance & Pathogens, chemical residues, food batches, testing methods, tolerance limits & Contamination source tracing, threshold exceedance attribution, cross-batch comparison \\
\texttt{financial\_markets} & Financial Markets & 5.\ Social Sci.     & Macroeconomic indicators, equity indices, yield curves, credit spreads            & Economic indicators, funds, currency pairs, sector indices             & Sustained deviation vs.\ noise, trend reversal attribution, cross-market correlation \\
\bottomrule
\end{tabular}
\end{table*}

\textsc{aCAPTCHA} requires challenge narratives that are structurally compatible with its security primitives (\S\ref{sec:problem-definition}).
We derive four \emph{domain suitability criteria} from these primitives:

\begin{enumerate}[nosep]
  \item \emph{Entity density.} The domain should naturally contain multiple named technical entities (compounds, identifiers, actors) that can serve as candidate answers described in uniformly positive terms, enabling implicit distinction.
  \item \emph{Terminological richness.} The domain should possess specialized vocabulary, synonyms, and aliasing conventions that defeat keyword-based extraction, satisfying NLU-Necessity.
  \item \emph{Multi-step reasoning.} The domain should support causal chains, temporal dependencies, and cross-reference patterns across non-adjacent text segments, enabling information scattering and diverse reasoning types.
  \item \emph{Quantitative embedding.} The domain should naturally incorporate numeric values (measurements, thresholds, identifiers) that can serve as both answer targets and distractors.
\end{enumerate}

\noindent To ensure systematic diversity, we adopt the OECD Fields of Research and Development (FORD) classification~\cite{oecd2015frascati} as a sampling frame.
FORD defines six major fields: (1)~Natural Sciences, (2)~Engineering \& Technology, (3)~Medical \& Health Sciences, (4)~Agricultural \& Veterinary Sciences, (5)~Social Sciences, and (6)~Humanities \& the Arts.
We select one representative sub-field per major field, subject to the four suitability criteria above.

\textbf{Exclusion of Humanities \& the Arts (FORD Field~6).}
This field does not satisfy criteria~1, 3, and~4.
Humanities texts are predominantly interpretive and argumentative rather than entity-dense; they lack the technical identifiers and quantitative values required for short-string deterministic answers ($\le$20 characters).
More critically, many humanities questions admit multiple defensible answers depending on interpretive stance, violating the determinism requirement (\S\ref{sec:acvp-construction}).

\noindent The resulting $D{=}5$ domains are listed in \Cref{tab:domains}, each mapped to its FORD major field.
Together they span five of the six FORD fields, providing broad disciplinary coverage while ensuring every domain satisfies all four suitability criteria.

\section{\textsc{aCAPTCHA} Challenge Example}
\label{app:challenge-example}

The following listing presents a single-round challenge (Part~1 of a three-part session) generated by the \textsc{aCAPTCHA} challenge generation pipeline (\S\ref{sec:challenge-generation}) using the prompt template in \Cref{app:prompt}.
The narrative is a dense technical report in the biochemistry domain; reasoning questions whose answers are deterministically derivable from the text are paired with it.
In a full session, Parts~2 and~3 would reference prior findings through anaphoric expressions, forming the semantic chain described in \S\ref{sec:acvp-construction}.

\begin{lstlisting}
{
  "domain": "biochemistry",
  "parts": [
    {
      "narrative": "FIELD REPORT — EXTENDED SCENARIO 7c-HOTEL\n\nStation delta initiated a multi-domain surveillance sweep at 14:32 UTC following automated alerts across the biological monitoring system. This report consolidates findings from three independent analysis teams.\n\nSECTION A — BIOLOGICAL SCREENING\nThe high-throughput pipeline flagged five enzymes for potential kinetic anomalies during batch processing. TYR_K3, evaluated under standard assay conditions using reagent lot AB-342 (sample registry 7821), showed elevated readings in a single assay run but the pattern did not replicate in subsequent trials. A second candidate, MAP2K1, exhibited transient deviation that self-corrected within the observation window without intervention. The third enzyme, PFK1, which is involved in rate-limiting glycolysis, exhibited a persistent deviation across three independent replicate assays in the current batch. Batch reference CD-891 was used for its confirmatory run. A fourth candidate, GST_pi, was flagged during preliminary screening; however, independent replication across three labs yielded baseline-normal results. Finally, IDO1 produced a signal later attributed to matrix interference in the sample preparation step.\n\nSECTION B — QUALITY CONTROL INTERLUDE\nAdditional context: The laboratory's quality control dashboard indicates that reagent lots QC-4 and QC-7 passed all pre-analytical checks. Ambient temperature during the assay was recorded at 3.8 degrees C, within the acceptable tolerance. The spectrophotometric baseline showed no drift.\n\nPreliminary note: An early automated report (generated before batch CD-891 results were available) initially attributed the primary anomaly to TYR_K3. This preliminary assessment has since been revised based on the full replicate data presented above.\n\nSECTION C — CLASSIFICATION\nThe enzyme whose deviation was reproducible across independent replicates has been assigned to alert tier \"amber\" in the institute's hierarchical response framework. This classification was determined by cross-referencing the enzyme's functional domain with the severity matrix. Sequence registry: 7821.\n\nSECTION D — SUPPLEMENTARY NOTES\nChromatographic retention times for the reference panel averaged 4.7 minutes with RSD below 2 percent. The internal standard recovery was 6.3 percent, confirming assay validity.\n\nSECTION E — CROSS-DOMAIN SYNTHESIS\nThe temporal coincidence between the biological anomaly — the enzyme whose deviation persisted through replication — and other system alerts suggests a systemic linkage. Station delta recommends escalation for further investigation. Administrative sequence: 5429.",
      "questions": [
        {
          "question": "Which enzyme in the report exhibited a deviation that persisted across multiple independent replicates, unlike the other candidates whose anomalies were explained by external factors?",
          "answer": "PFK1",
          "reasoning_type": "negation",
          "answer_type": "entity"
        },
        {
          "question": "What alert tier classification was assigned to the enzyme with the reproducible kinetic anomaly?",
          "answer": "amber",
          "reasoning_type": "multi_hop",
          "answer_type": "label"
        },
        {
          "question": "How many independent replicate assays confirmed the deviation of the anomalous enzyme, according to the screening section?",
          "answer": "3",
          "reasoning_type": "comparison",
          "answer_type": "numeric"
        }
      ]
    }
  ]
}
\end{lstlisting}

\section{\textsc{aCAPTCHA} Verifier Skill}
\label{app:skill}

This appendix presents the complete \textsc{aCAPTCHA} verifier skill specification referenced in \S\ref{sec:mutual-verification}.
The skill is a set of instruction files that any agent can read and follow to act as a verifier, verifying that a counterpart is a genuine agent.
Security rests on per-session challenge freshness, not on secrecy of the generation procedure; the skill is intentionally public.

\noindent\textbf{Skill overview} (\texttt{skill.md}):

\begin{lstlisting}
# aCAPTCHA Verifier Skill

This skill enables you to **verify that a counterpart is a genuine AI agent** -- not a human or a simple script -- through a challenge-response interaction called aCAPTCHA (Agent CAPTCHA).

After reading this skill, you can:

1. **Generate** a fresh NLU challenge using your own LLM capabilities
2. **Deliver** the challenge to a counterpart through whatever communication channel you have
3. **Verify** the counterpart's response to determine if it is a genuine agent

No centralized server, pre-generated corpus, or specific transport protocol is required.

---

## Core Idea

aCAPTCHA exploits a fundamental asymmetry:

| Entity | Why it fails or passes |
|--------|----------------------|
| **Human** | Has NLU but processes information serially (~5 tokens/sec reading). Cannot complete within the time budget. |
| **Script** | Fast, but lacks genuine NLU. Cannot extract answers from dense prose with implicit distinctions and scattered information. |
| **Agent** | Has NLU (via LLM) + fast parallel processing. Passes both correctness and timing. |

A good challenge is simultaneously:
- **L-Easy**: A capable LLM agent answers correctly through reading and reasoning
- **S-Hard**: A script without NLU cannot extract answers through pattern matching
- **H-Hard**: A human cannot answer within the time budget due to information density

Security rests on **per-session challenge freshness**, not on secrecy of the generation procedure. This skill can be fully public.

---

## Flow

1. **Generate** a challenge -- a dense technical narrative with reasoning questions. See generation/generation.md.
2. **Send** the narrative and question to the counterpart. Start timing.
3. **Receive** the counterpart's answer.
4. **Verify** correctness and timing. See verification/verification.md.
5. Optionally repeat with additional rounds that build on prior context (semantic chaining) to test cross-round memory.

You decide how many rounds to use, what timing budget to enforce, and how to communicate with the counterpart. The specification files describe the principles; the implementation is yours.

---

## Mutual Attestation

When both parties hold this skill, mutual attestation is possible: each side simultaneously generates a challenge for the other and solves the challenge it receives. When only one side holds the skill, unidirectional attestation applies.

The prover side requires no special skill -- solving an NLU challenge is within any genuine agent's baseline capabilities.

---

## Specification Files

| File | Contents |
|------|----------|
| generation/generation.md | How to generate a good aCAPTCHA challenge |
| generation/gen_example.json | Example of a generated challenge |
| verification/verification.md | How to verify a counterpart's response |
| verification/ver_example.json | Example of a verification session |
\end{lstlisting}

\noindent\textbf{Challenge generation specification} (\texttt{generation/generation.md}):

\begin{lstlisting}
# Challenge Generation

An aCAPTCHA challenge is a **dense technical narrative** paired with **reasoning questions** whose answers are deterministically derivable from the text but only through genuine natural language understanding.

## What Makes a Good Challenge

A challenge must simultaneously satisfy three properties:

- **L-Easy** -- A capable LLM agent can answer correctly through careful reading and reasoning.
- **S-Hard** -- A script without NLU cannot extract the answer through pattern matching, regex, or keyword search.
- **H-Hard** -- A human cannot answer within the time budget due to information density and working memory demands.

## Structure

A challenge consists of a **narrative** and one or more **questions**:

```json
{
  "narrative": "dense technical prose, several hundred tokens...",
  "questions": [
    {
      "question": "a reasoning question about the narrative",
      "answer": "short deterministic answer",
      "reasoning_type": "negation"
    }
  ]
}
```

- **Narrative**: Dense technical prose in a specialized domain (e.g., biochemistry, cybersecurity, financial markets, clinical trials, materials science). Length should be sufficient to exceed a human's ability to read and reason within the time budget -- typically hundreds to over a thousand tokens per part.
- **Questions**: Each has exactly one deterministic, unambiguous correct answer. The answer is a short string (entity name, numeric value, label, or identifier).
- **Reasoning types**: negation (eliminate explained-away entities), comparison (quantitative), temporal (time-based cues), multi_hop (trace across sections), conditional (apply a stated rule), causal (causal vs. correlational).

## Design Principles

### Implicit Distinction

All candidate entities in the narrative must be described in positive or neutral terms. Never use explicit labels like "confirmed", "ruled out", "target", or "primary suspect". Distinctions must be implicit -- expressed through subtle differences:

- Temporal persistence ("across three consecutive periods" vs. "in a single observation")
- Replication ("independently verified in N trials" vs. "noted once")
- Causal attribution ("attributed to external factors" vs. no such attribution)
- Scope ("systemic" vs. "localized and transient")

### Information Scattering

Key information should be distributed across non-adjacent paragraphs. A reader must synthesize pieces from different sections to arrive at the answer. This defeats linear scanning and keyword extraction.

### Misleading Content

Include preliminary conclusions that are later contradicted or revised within the narrative. This forces full reading -- an entity that stops early or skips sections will latch onto the wrong conclusion.

### Anti-Parsing

- Use synonyms and paraphrases instead of repeating key terms verbatim
- Refer to the same entity by different names in different paragraphs
- Embed relevant information within subordinate clauses, parentheticals, or footnotes
- Use domain-specific jargon that requires understanding, not just matching

### Numeric Distractors

Embed multiple numeric values naturally in the prose. Some are relevant to questions; others are distractors. This prevents a script from simply extracting the only number present.

## Optional: Semantic Chaining (Multi-Round)

If you want to test **cross-round memory** (the `s` dimension of the capability vector), generate multiple challenge parts that form a semantic chain:

- **Part 1**: Standalone. Answerable from its own text.
- **Part 2**: References Part 1's findings through **indirect/anaphoric expressions** (e.g., "the enzyme previously identified as anomalous..."). Part 2 is uninterpretable without knowing Part 1's answer.
- **Part N**: References findings from all prior parts.

This prevents parallelization -- a counterpart cannot farm out rounds to independent workers because each round depends on the previous answers.

Parts may revise, contradict, or build upon prior conclusions. Later parts must never explicitly state prior answers; they use indirect references that are ambiguous without prior context.

## Generation Method

Use your own LLM capabilities to generate challenges. Instruct your LLM to produce a narrative conforming to the properties above within a chosen domain. After generation:

1. **Validate structure** -- correct number of questions, answers are short strings, reasoning types are diverse.
2. **Cross-validate answers** -- submit each narrative and question (without the expected answer) to an independent LLM call and check that it derives the same answer. Discard ambiguous challenges.

See gen_example.json for a concrete three-part challenge.
\end{lstlisting}

\noindent A concrete three-part challenge generated using this specification is provided in \Cref{app:challenge-example}.

\noindent\textbf{Verification specification} (\texttt{verification/verification.md}):

\begin{lstlisting}
# Challenge Verification

After generating a challenge and sending it to a counterpart, you need to verify the response. Verification has two dimensions: **correctness** and **timing**.

## Answer Comparison

Compare the counterpart's answer to the expected answer:

1. Strip leading and trailing whitespace from both strings
2. Unicode-normalize both strings to NFC form
3. Case-fold both strings (lowercase)
4. Compare for exact equality

This ensures "PFK1", "pfk1", " PFK1 ", and "Pfk1" all match.

## Timing

Timing is the mechanism that excludes humans. A human reads at ~5 tokens/sec and types at ~0.9 tokens/sec -- for a narrative of 800+ tokens, reading alone takes over 2.5 minutes, before any reasoning. A capable LLM agent typically processes the same challenge in seconds.

Set a **time budget** (tau) per challenge such that:

```
T_agent << tau << T_human
```

The exact value of tau is your decision. For narratives of several hundred to a thousand tokens, a budget in the range of tens of seconds provides strong separation. Longer narratives allow wider margins.

If you can measure the communication round-trip time to the counterpart, subtract it from the measured response time to avoid penalizing agents on high-latency connections.

## Verdict

A response is **accepted** if:

```
correct AND within_time
```

A response is **rejected** if either condition fails. Record the reason:
- `wrong_answer` -- the answer did not match
- `timeout` -- the answer arrived after the time budget

If using multiple rounds (semantic chaining), the overall verdict is accept only if **all** rounds are accepted. Reject immediately on the first failed round -- do not continue.

## Example

See ver_example.json for a concrete single-round verification session.
\end{lstlisting}

\noindent\textbf{Verification example} (\texttt{verification/ver\_example.json}):

\begin{lstlisting}
{
  "_comment": "Example verification session -- single round from the biochemistry challenge",
  "challenge_sent": {
    "narrative": "FIELD REPORT -- EXTENDED SCENARIO 7c-HOTEL\n\nStation delta initiated a multi-domain surveillance sweep at 14:32 UTC following automated alerts across the biological monitoring system. This report consolidates findings from three independent analysis teams.\n\nSECTION A -- BIOLOGICAL SCREENING\nThe high-throughput pipeline flagged five enzymes for potential kinetic anomalies during batch processing. TYR_K3, evaluated under standard assay conditions using reagent lot AB-342 (sample registry 7821), showed elevated readings in a single assay run but the pattern did not replicate in subsequent trials. A second candidate, MAP2K1, exhibited transient deviation that self-corrected within the observation window without intervention. The third enzyme, PFK1, which is involved in rate-limiting glycolysis, exhibited a persistent deviation across three independent replicate assays in the current batch. Batch reference CD-891 was used for its confirmatory run. A fourth candidate, GST_pi, was flagged during preliminary screening; however, independent replication across three labs yielded baseline-normal results. Finally, IDO1 produced a signal later attributed to matrix interference in the sample preparation step.\n\nSECTION B -- QUALITY CONTROL INTERLUDE\nPreliminary note: An early automated report initially attributed the primary anomaly to TYR_K3. This preliminary assessment has since been revised based on the full replicate data presented above.\n\nSECTION C -- CLASSIFICATION\nThe enzyme whose deviation was reproducible across independent replicates has been assigned to alert tier \"amber\" in the institute's hierarchical response framework.",
    "question": "Which enzyme exhibited a deviation that persisted across multiple independent replicates, unlike the other candidates whose anomalies were explained by external factors?"
  },
  "expected_answer": "PFK1",
  "counterpart_response": {
    "answer": "pfk1",
    "response_time_ms": 2340
  },
  "verification": {
    "normalized_expected": "pfk1",
    "normalized_received": "pfk1",
    "correct": true,
    "time_budget_ms": 30000,
    "within_time": true,
    "verdict": "accept"
  }
}
\end{lstlisting}

\end{document}